\documentclass[12pt]{amsart}

\usepackage[english]{babel}
\usepackage[utf8]{inputenc}
\usepackage[T1]{fontenc}
\usepackage{amsthm}
\usepackage{amsmath}
\usepackage{amssymb}
\usepackage{mathrsfs}
\usepackage{setspace}
\usepackage{graphicx}
\usepackage{float}
\usepackage[top=2cm,bottom=2cm,left=2cm,right=2cm]{geometry}
\usepackage{multirow}
\usepackage{slashbox}%
\usepackage{stmaryrd}
\usepackage{soul}
\usepackage{color}
\usepackage{verbatim}
\usepackage{alltt}
\usepackage{wrapfig}
\usepackage{pifont}
\usepackage{array}
\usepackage{colortbl}
\usepackage{ifmtarg}
\usepackage{hyperref}
\usepackage{thmtools}
\usepackage{dsfont}
\usepackage{mathabx}
\usepackage{relsize,exscale}
\numberwithin{equation}{section} 
\usepackage{blindtext}

\declaretheorem[title=Definition,numberwithin=section]{definition}
\declaretheorem[numberwithin=section]{proposition}
\declaretheorem[title=Corollary,numberwithin=section]{corollary}
\declaretheorem[title=Theorem,numberwithin=section]{theorem}
\declaretheorem[title=Lemma,numberwithin=section]{lemma}

\declaretheorem[title=Assumption,numberwithin=section]{assumption}
\declaretheorem[title=Remark,numberwithin=section]{remark}

\newcommand{\R}{\ensuremath{\mathbb{R}}\xspace}
\newcommand{\Z}{\ensuremath{\mathbb{Z}}\xspace}

\newcommand{\N}{\ensuremath{\mathbb{N}}\xspace}
\newcommand{\C}{\ensuremath{\mathbb{C}}\xspace}

\newcommand{\supp}{\text{supp}}

\begin{document}

\title{Thermodynamic limit of the Pieces' Model}
\author{Vadim Ognov}
\address{Institut de Mathematiques de Jussieu - Paris Rive Gauche\\ Sorbonne Université \\ 75005 Paris, France}
\email{vadim.ognov@imj-prg.fr}

\begin{abstract}
We study the ground states of the pieces' model in the Fermi-Dirac statistics in the thermodynamic limit. In other words, we consider the minimizing configurations of $ n $ interacting fermions in an interval $ \Lambda $ divided into pieces by a Poisson point process, when $ \frac{n}{\vert \Lambda\vert}\to \rho>0 $ as $ \vert \Lambda \vert \to \infty $. We notice that a decomposition into groups of pieces arises from the hypothesis of finite-range pairwise interaction. Under assumptions of convexity and non-degeneracy of the subsystems, we get an almost complete factorization of any ground state. This method applies at least for groups comprising one or two particles. It improves the expansion of the thermodynamic limit of the ground state energy per particle up to the error $ O(\rho^{2-\delta}) $, with $ 0<\delta<1 $ (see \cite{Klopp2020}). It also provides an approximate ground state for the pieces' model. 
\end{abstract}

\maketitle

\section{Introduction}

One-dimensional many-body localization is a non-trivial topic for both condensed matter physicists and spectral theory mathematicians. At large disorder, one expects that quantum systems with interaction do not thermalize and that they exhibit a kind of localization \cite{Alet2018}. Some papers tackle this phenomenon for a finite number of particles and an infinite interval \cite{Beaud2018} \cite{Elgart2018}. However, from a physical perspective, the appropriate scope would be to consider a number of particles that increases proportionally with the size of the interval. This regime is called the \textit{thermodynamic limit}. 

Published in 2012, a paper of Veniaminov proved the existence of the thermodynamic limit of the ground state energy per particle for a class of disordered quantum systems \cite{Veniaminov2012}. This result applies in particular to the \textit{pieces' model} which is a refined version of the Luttinger-Sy model, introduced in 1973 \cite{Luttinger1973}.  Without interaction, the ground state is given by minimizing the distribution of $ n $ particles among the partition of the large interval $ \Lambda $ into pieces by the Poisson point process. Because of this explicit solution and since the original paper, the pieces' model has been studied to understand the Bose-Einstein condensation of free or interacting bosons \cite{Lenoble2006}\cite{Kerner2019a}\cite{Kerner2019} \cite{Kerner2021}.

In this article, we focus on the pieces' model in the Fermi-Dirac statistics, i.e for indistinguishable particles. Our work is inspired by the paper of Klopp and Veniaminov \cite{Klopp2020}. Let $ \rho>0 $ be the density of particles, i.e the limit of the ratio $ \frac{n}{\vert \Lambda \vert} $. Klopp and Veniaminov expand the thermodynamic limit of the ground state energy per particle up to the error $ O\big(-\rho \log(\rho)^{-3}\big) $. We give an expansion up to the error $ O\big(\rho^{2-\delta}\big) $, for any $ \delta\in (0,1) $, in case of finite-range interactions. We also provide a natural characterization of the ground states. The next step would be to use our results to express some indicators of the many-body localization.

Let us now briefly describe our method. In the free case, the minimizing configuration of particles is such that the energy produced by any particle is less than the \textit{Fermi energy} $ E_{\rho} $. It yields that, in the ground state, the pieces with length below $ l_{\rho}=\pi E_{\rho}^{-1/2} $ are empty. Similarly, in the interacting case, under the assumptions of a pairwise potential $ U $ with compact support and a density of particles $ \rho $ small enough, the pieces with length below $ l_{\rho,U} $ are empty for any ground state. So, the random background reduces to a compilation of groups of pieces, that we call \textit{chains}, such that a particle belonging to a chain cannot interact with a particle living outside this chain. This structure is therefore similar to the one of the free system if the chains replace the pieces. Our problem turns into finding a minimizing distribution of $ n $ particles among the chains. Without interaction, given any piece, the energy as a function of the number of particles is convex. This property allows to get the ground state inductively. Does this statement hold for any chain in the presence of interactions? Unfortunately we did not solve this question. We bypass this issue noticing that, due to the nature of the Poisson point process, large chains do not contribute much to the total energy. The ground state energy per particle is mostly, i.e up to our error term, given by isolated fermions and isolated pairs of fermions lying in one or two pieces. For these simple subsystems, the energies are convex and we can compare them quite precisely. Then, we distribute by induction the particles among these chains. We prove that the corresponding state approximates any ground state in the thermodynamic limit.

The paper is organized as follows. In Section 2, we present the model and we sketch our method to get an expansion of the ground state energy per particle up to any order $ O(\rho^{p-\delta}) $, $ p\geq 2 $ and $ 0<\delta<1 $, under strong assumptions. In Section 3, we state our results for $ p=2 $ without proof. Section 4 rigorously develop the splitting into chains, including its limits. Section 5 is devoted to the detailed study of chains comprising at most two particles. It also contains the proofs of our main propositions. We gather other results in the Appendix.

$ \\ $

\section{Model and first observations}
\subsection{The pieces' model for Fermi-Dirac statistics}

Let $ X(w)=(x_{n}(w))_{n\in \Z} $ be a Poisson point process on $ \R $ of intensity $ 1 $. Recall that the probability that a Borel set $ \Lambda\subset \R $ contain exactly $ k $ points is
\[\mathbb{P}\Big(\#\big(X(w)\cap \Lambda\big)=k\Big)=\frac{\vert \Lambda \vert^{k}}{k!}e^{-\vert \Lambda \vert} \]
and for two disjoints Borel sets $ \Lambda_{1},\Lambda_{2}\subset \R $, the events $ \{X(w)\cap \Lambda_{1}=k_{1}\} $ and $ \{ X(w)\cap \Lambda_{2}=k_{2}\} $ are independent.

 For $ L>0 $ we set $ \Lambda=[0,L] $. We assume that $ x_{0}(w)=0 $ and we denote $ m(w)=\#\big(X(w)\cap \Lambda\big) $.  By a large deviation principle, when $ L $ is large, with probability $ 1-O(L^{-\infty}) $, $ m(w)=L+O(L^{\frac{2}{3}}) $,. For $ i\in\llbracket 1, m(w) \rrbracket $, the \textit{$ i $-th piece} is the interval $ \Delta_{i}(w)=[x_{i-1}(w),x_{i}(w)] $. 
 
 On $ \mathfrak{H}(\Lambda)=L^{2}(\Lambda) $, we set the following one-particle random operator
\begin{equation}
h_{w}(\Lambda)=\bigoplus_{k=1}^{m(w)}\bigg(-\frac{d^{2}}{dx^{2}}^{D}_{\vert \Delta_{k}(w)}\bigg)
\end{equation} 
where $ D $ stands for Dirichlet boundary conditions.

Now,  we consider $ n $ particles in the disordered background given by $ h_{w}(\Lambda) $ combined to a pairwise repulsive interaction. Using the statistic of Fermi-Dirac, the \textit{$ n $-particle space} on $ \Lambda $ is
\begin{equation}
 \mathfrak{H}^{n}(\Lambda)=\bigwedge_{i=1}^{n}\mathfrak{H}(\Lambda).
 \end{equation}
Then, for $ n\geq 2 $, the \textit{pieces' model} is the random operator given by 
\begin{equation}\label{H}
H^{U}_{w}(\Lambda,n)=\sum_{i=1}^{n}\bigg(\bigotimes_{j=1}^{i-1}\mathbf{1}_{\mathfrak{H}(\Lambda)}\bigg) \otimes h_{w}(\Lambda) \otimes \bigg(\bigotimes_{j=1}^{n-i} \mathbf{1}_{\mathfrak{H}(\Lambda)}\bigg)+W_{n} \qquad \text{on } \mathfrak{H}^{n}(\Lambda) 
\end{equation}
where $ W_{n} $ is the multiplication operator
\begin{equation}\label{W}
W_{n}(x_{1},\dots,x_{n})=\sum_{i<j}U(x_{i}-x_{j})
\end{equation} 
and $ U:\R\longrightarrow \R $ satisfies the following assumption.
\begin{assumption}\label{hypoU}
The function $ U:\R\rightarrow \R $ is nonnegative, even, bounded and compactly supported.
\end{assumption}
Under Assumption \ref{hypoU}, the operator $ H^{U}_{w}(\Lambda,n) $ is well-defined on $ \mathcal{D}_{w}(\Lambda,n) $ given by 
$$ \mathcal{D}_{w}(\Lambda,n)=\mathcal{C}^{\infty}_{0}\bigg(\Big( \bigcup_{k=1}^{m(w)}]x_{k-1},x_{k}[\Big)^{n}\bigg)\cap \mathfrak{H}^{n}(\Lambda) $$
 and it is nonnegative. Using perturbation theory (see e.g Chapter 6 \cite{Teschl2014}), one proves that $ H^{U}_{w}(\Lambda,n) $ is essentially self-adjoint on $ \mathfrak{H}^{n}(\Lambda) $ and it has pure spectrum. Let $ E_{w}^{U}(\Lambda,n) $ be the ground state energy of $ H_{w}^{U}(\Lambda,n) $.

\begin{definition}\label{thermolim}
The limit $  \big\{ L\to +\infty $, $ \frac{n}{L}\to \rho  \big\} $ is called the  \textbf{thermodynamic limit}. The constant $ \rho $ is the \textbf{density of particles} per unit of volume.
\end{definition}
In \cite{Klopp2020}, Klopp and Veniaminov proved that, even under weaker assumptions on $ U $, the thermodynamic limit of $ n^{-1}E_{w}^{U}(\Lambda,n) $ exists $ \mathbb{P} $-almost surely and in $ L^{1}(\mathbb{P}) $. In this paper, we give an expansion of this limit.

  \subsection{The free operator} We denote by $ H_{w}^{0}(\Lambda,n) $ the free operator and by $ E_{w}^{0}(\Lambda,n) $ its ground state energy. One can give quite explicitly the thermodynamic limit of the ground state energy per particle
\begin{equation}
\mathcal{E}^{0}(\rho):=\lim_{\substack{L\to +\infty\\ \frac{n}{L}\to \rho}}\frac{E_{w}^{0}(\Lambda,n)}{n}.
\end{equation}
The ground state energy $ E^{0}_{w}(\Lambda,n) $ is exactly the sum of the $ n $ first eigenvalues of $ h_{w}(\Lambda) $. But, since its eigenvalues only depend on the lengths of the pieces and the statistical distribution of these lengths is known, the pieces' model admits an explicit \textit{integrated density of states} (see Proposition 2.6 \cite{Klopp2020} or Proposition 3.2 \cite{Lenoble2006}). One computes
\begin{equation}\label{IDS}
N(E):=\lim_{L\to \infty}\frac{\# \Big\{\text{eigenvalues of } h_{w}(\Lambda) \text{ in } (-\infty,E] \Big\}}{L}=\frac{e^{-\frac{\pi}{\sqrt{E}}}}{1-e^{-\frac{\pi}{\sqrt{E}}}}1_{E\geq 0}
\end{equation}
Let the \textit{Fermi energy} $ E_{\rho} $ be the unique solution of $ N(E)=\rho $. Then, one deduces
\begin{equation}\label{freeE}
\mathcal{E}^{0}(\rho)=\frac{1}{\rho}\int_{-\infty}^{E_{\rho}}E\,dN(E). 
\end{equation}
We refer to Theorem 5.14 \cite{Veniaminov2012} for the proof.

  \subsection{The approach in term of occupations}
From now on, we drop the $ "w" $ index. Unlike the free operator, one cannot express the ground state energy of the pieces' model with interactions by using the spectral decomposition of the one-particle operator. However, in both cases, one can talk about the number of particles in a given piece. The $ n $-particle space admits the decomposition
\begin{equation}\label{decomposition}
\mathfrak{H}^{n}(\Lambda)=\bigoplus_{Q\in \N^{m}, \, \vert Q \vert_{1}=n}\mathfrak{H}_{Q}(\Lambda) \qquad \text{ with }  \quad \mathfrak{H}_{(q_{i})_{1\leq i\leq m}}(\Lambda)=\bigwedge_{i=1}^{m}\bigg(\bigwedge_{j=1}^{q_{i}}L^{2}(\Delta_{i})\bigg).
\end{equation}

\begin{definition}
An \textbf{occupation} is a multi-index $ Q=(q_{i})_{1\leq i\leq m } $ of norm equal to $ n $.
\end{definition}
In \cite{Klopp2020}, Klopp and Veniaminov proved that the decomposition (\ref{decomposition}) is invariant under the action of $ H^{U}(\Lambda,n) $. For a fixed occupation $ Q $, let $ H^{U}(\Lambda,n,Q) $ be the restriction of $ H^{U}(\Lambda,n) $ to the subspace $ \mathfrak{H}_{Q}(\Lambda) $. Then, the ground state $ \psi^{U}(\Lambda,n,Q) $ of  $ H^{U}(\Lambda,n,Q) $ is non-degenerate and it has exactly $ q_{i} $ particles in the piece $ \Delta_{i} $ for all $ i\in \llbracket 1,m\rrbracket $.

In the free case, it yields that, for a given occupation $ Q $, the ground state energy of $ H^{0}(\Lambda,n,Q) $ satisfies 
\begin{equation}\label{sommelibre}
E^{0}(\Lambda,n,Q)=\sum_{i=1}^{m}E^{0}(\Delta_{i},q_{i})
\end{equation}
where we denote $ E^{0}(\Delta,k) $ the ground state energy for $ k $ non-interacting fermionic particles in the piece $ \Delta $. Each particle lies in a Dirichlet Laplacian background in $ \Delta $. The minimum of $ E^{0}(\Lambda,n,Q) $ over all the occupations is the ground state energy of $ H^{0}(\Lambda,n) $. Remark that $ E^{0}(\Delta,k) $ is the sum of the $ k $ first eigenvalues of the operator  $ h_{\Delta}=-\frac{d^{2}}{dx^{2}}^{D}_{\vert \Delta} $. So, the map $ k\ \rightarrow E^{0}(\Delta,k) $  is strictly convex on $ \N $. By Lemma \ref{convexsomme}, the ground state energy $ E^{0}(\Lambda,n) $ is given by the sum of the $ n $ smallest elements of the set $ \Gamma^{0}=\big\{ E^{0}(\Delta_{i},k+1)-E^{0}(\Delta_{i},k), \, i\in \llbracket 1,m \rrbracket , \, k\in \N \big\} $.

However, note that the set $ \Gamma^{0} $ is equal to the set of all the eigenvalues of the one-particle operator $ h(\Lambda) $. Then the counting function of $ \Gamma^{0} $,
\begin{equation}\label{N0}
N^{0}(E):=\lim_{L\to +\infty}\frac{\# \Big(\Gamma^{0}\cap (-\infty,E]\Big) }{L},
\end{equation}
 is well-defined and it is equal to the integrated density of state of $ h(\Lambda) $. Thus, we recover the formula (\ref{freeE}).
 
$ \\ $

From now on, we restrict to finite-range interactions.
\begin{assumption}\label{finiterange}
Let $ s(U) $ be the support of the function $ U $ and
\begin{equation}
M=\sup_{x,y\, \in \, s(U)}\vert x-y \vert
\end{equation}
The length $ M $ is independent of $ \rho $.
\end{assumption}
The following lemma is crucial for our analysis. 
\begin{lemma}\label{lrhoU}
Let $ \Psi^{U}(\Lambda,n) $ to be a ground state of $ H^{U}(\Lambda,n) $. For $ n $ and $ L $ large enough, with probability $ 1-O(L^{-\infty}) $, there exists a minimal length $ l_{\rho,U}=-\log\big(\frac{\rho}{1+\rho}\big)-(4M+6)\rho $ such that
\begin{center}
\textit{If a piece $ \Delta_{i} $ satisfies $ \vert \Delta_{i} \vert<kl_{\rho,U} $, $ k\in \N $, then, for every occupation $ Q $,}
$$ \Big(P_{Q}\Psi^{U}(\Lambda,n)\neq 0 \Big) \Rightarrow \Big(q_{i}\leq k-1 \Big) $$
\end{center}
where $ P_{Q} $ is the orthogonal projector on $ \mathfrak{H}_{Q}(\Lambda) $.
\end{lemma}
So, given a piece, the number of particles in this piece is bounded uniformly for any ground state. In particular, the pieces of length up to $ l_{\rho,U} $ are empty for any ground state.

We will use the term $ \textit{chain} $ to refer to a group of pieces of length greater than $ l_{\rho,U} $ with gaps of length smaller than $ M $. Let $ \mathcal{P} $ to be the set of chains.  Using the notations of Lemma \ref{lrhoU}, for any occupation $ Q $ such that $ P_{Q}\Psi^{U}_{\omega}(\Lambda,n)\neq 0 $, the ground state energy of $ H^{U}(\Lambda,n, Q) $ satisfies
\begin{equation}\label{sommeU}
E^{U}(\Lambda,n,Q)=\sum_{I\in \mathcal{P}}F^{U}\big(I,\kappa_{I}(Q)\big)
\end{equation}
where $ \kappa_{I}(Q) $ is the number of particles in the chain $ I $ and $ F^{U}(I,\kappa) $ is the smallest energy produced by $ \kappa $ particles in  $ I $. Each particle lies in a Dirichlet Laplacian background for some piece of $ I $ and it is eventually submitted to the repulsive pairwise interaction $ U $.

One should think of Equation (\ref{sommeU}) as a counterpart to Equation (\ref{sommelibre}) where each chain stands for an occupied piece in the free case. If one could prove the convexity of every map $ \kappa \rightarrow F^{U}(I,\kappa) $ then by Lemma \ref{convexsomme}, the ground state energy $ E^{U}(\Lambda,n) $ would be given by the sum of the $ n $ smallest elements of the set $ \Gamma=\{F^{U}(I,\kappa+1)-F^{U}(I,\kappa), \, I \text{ chain},\, \kappa\in \N\}  $.

For $ \kappa\ge 0 $, the \textit{$ (\kappa +1) $-th energy level of the chain $ I $} is given by 
\begin{equation}
f^{U}(I,\kappa+1)=F^{U}(I,\kappa+1)-F^{U}(I,\kappa). 
\end{equation}
It represents the smallest amount of energy that appears if one adds a particle to a minimizing configuration of $ \kappa $ particles in $ I $. From the above discussion, one would like to use that, for every chain, $ \kappa \rightarrow f^{U}(I,\kappa) $ is increasing. Using the perturbation methods, we fail to prove such a statement. However it seems relevant to search for results in case of monotony for small chains and/or for few particles.

More precisely, let $ p\geq 2 $ and $ \mathcal{P}_{p} $ be the set of chains each of which carries at most $ p $ particles for any ground state, and $ \Gamma_{p} $ be the set of the $ p $ lowest energy levels of every chain that belongs to $ \mathcal{P}_{p} $, meaning that
\begin{equation}
\Gamma_{p}=\Big\{f^{U}(I,\kappa), \, I\in \mathcal{P}_{p}, \, \kappa\leq p \Big\}.
\end{equation}
Assume that
\begin{equation}\label{convex0}
\forall I  \in \mathcal{P}_{p},  \quad \forall \kappa\leq p-1 , \qquad f(I,\kappa)< f(I,\kappa+1).
\end{equation}
Set $ \delta\in (0,1) $. By Lemma \ref{lrhoU} and by statistical distribution of the pieces (see Proposition \ref{nb3}), one proves that, for any ground state, the number of particles in  $ ^{c}\mathcal{P}_{p} $, the complement of $ \mathcal{P}_{p} $, is of order $ O(n\rho^{p-\delta}) $. One also controls the contribution of these particles to the ground state energy with a bound of order  $ O(n\rho^{p-\delta}) $. Then, up to an error $ O(n\rho^{p-\delta}) $, the ground state energy $ E^{U}(\Lambda,n) $ is given by the sum of the $ n $ smallest elements of $ \Gamma_{p} $. Let $ N^{U}_{p} $ be the counting function of $ \Gamma_{p} $, meaning that
\begin{equation}
N^{U}_{p}(\lambda):=\lim_{L\to +\infty}\frac{\# \Big(\Gamma_{p}\cap (-\infty,\lambda]\Big) }{L}.
\end{equation}
Using $ N^{U}_{p} $ as a counterpart to $ N^{0} $ (see (\ref{N0})), one should get an approximation of the thermodynamic limit of the ground state energy per particle $ \mathcal{E}^{U}(\rho) $ up to an error $ O(\rho^{p-\delta}) $.

$ \\ $

\section{Main Results}\label{mainresult}

 Since the interaction is repulsive, Assumption (\ref{convex0}) is always true for $ p=2 $. Following the above discussion, we study this case in depth. In the set $ \mathcal{P}_{2} $, a chain is either a single piece with at most two particles, or a pair of pieces with at most one particle in each piece.

Klopp and Veniaminov proved a result about the ground state energy of two interacting particles in a single piece.
\begin{proposition}\label{Esolo}\cite{Klopp2020}
Under Assumption \ref{hypoU}, for $ l>0 $, consider the operator
\begin{equation}\label{ope1}
\bigg(-\frac{d^{2}}{dy^{2}}^{D}_{\vert [0,l]}\bigg)\otimes \mathbf{1}_{L^{2}([0,l])} +\mathbf{1}_{L^{2}([0,l])}\otimes \bigg(-\frac{d^{2}}{dx^{2}}^{D}_{\vert [0,l]}\bigg) + U(x-y) \qquad \text{ on } L^{2}\big([0,l]\big)\wedge L^{2}\big([0,l]\big) 
\end{equation}
Then, for large $ l $, the ground state energy $ E^{U}\big([0,l],2\big) $ admits the following expansion
\begin{equation}
E^{U}\big([0,l],2\big)=\frac{5\pi^{2}}{l^{2}}+\frac{\gamma}{l^{3}}+o(l^{-3})
\end{equation}
with $ \gamma>0 $ when $ U\neq 0 $.
\end{proposition}

In the Appendix, we prove an analogue of Proposition \ref{Esolo} for the ground state energy of two interacting particles in two distinct pieces.
\begin{proposition}\label{Epair}
Under Assumption \ref{hypoU}, for $ l>0 $, $ d\geq0 $ and $ a>1 $, consider the operator
\begin{equation}\label{ope2}
\bigg(-\frac{d^{2}}{dy^{2}}^{D}_{\vert [-al,0]}\bigg)\otimes \mathbf{1}_{L^{2}([d,d+l])} +\mathbf{1}_{L^{2}([-al,0])}\otimes \bigg(-\frac{d^{2}}{dx^{2}}^{D}_{\vert [d,d+l]}\bigg) + U(x-y) \quad \text{ on } L^{2}\big([-al,0]\big)\otimes L^{2}\big([d,d+l]\big) 
\end{equation}
Then, for $ d\geq 0 $ and large $ l>0 $, the ground state energy $E^{U}\Big(\big\{[-al,0],[d,d+l]\big\},(1,1)\Big) $ admits the following expansion
\begin{equation}\label{E2}
 E^{U}\Big(\big\{[-al,0],[d,d+l]\big\},(1,1)\Big)=\Big(\frac{\pi^{2}}{l^{2}}+\frac{\pi^{2}}{(al)^{2}}\Big)+\frac{\sigma(d)}{a^{3}l^{6}}\Big(1+o(1)\Big)
\end{equation}
with $ \sigma(d) $ a positive function that vanishes for $ d>\text{diam}(\supp(U)) $.
\end{proposition}

We now state our theorem.

\begin{theorem}\label{mainth}
Under Assumption \ref{hypoU} and Assumption \ref{finiterange}, let $ M=\text{diam}(\supp(U)) $ and $ l_{\rho,U}>0 $ be the minimal length defined in Lemma \ref{lrhoU}. Consider, on $ (0,+\infty) $, the application
\begin{align}
\mathcal{J}(\lambda)&=\big(1-Me^{-l_{\rho,U}}\big)^{2}\bigg(\int_{\mathcal{D}_{1}(\lambda)}f^{U}([0,u],1)e^{-u}\, du +\int_{\mathcal{D}_{2}(\lambda)}f^{U}([0,u],2)e^{-u}\,du \nonumber\\
&\quad +\int_{0}^{M}\int_{\mathcal{D}_{3}(\lambda)}2 e^{-(u+v)}f^{U}(\{[-u,0],[t,v+t]\},1)\,dtdudv  \nonumber \\
&\quad +\int_{0}^{M}\int_{\mathcal{D}_{4}(\lambda,t)}2e^{-(u+v)}f^{U}\big(\{[-u,0],[t,v+t]\},2\big)\,dtdudv  \bigg)\nonumber
\end{align}
where $ f^{U}(I,1) $ (resp. $ f^{U}(I,2) $) is the first (resp. second) energy level of the chain $ I $,
\begin{align*}
\mathcal{D}_{1}(\lambda)=\Big[\frac{\pi}{\sqrt{\lambda}}, 3l_{\rho,U}\Big]&, \qquad \mathcal{D}_{3}(\lambda)=\bigg\{(x,y)\in \big[l_{\rho,U},2l_{\rho,U}\big]^{2}, y\geq \max\Big(x,\frac{\pi}{\sqrt{\lambda}}\Big)\bigg\} \\
\mathcal{D}_{2}(\lambda)=\Big[\frac{2\pi}{\sqrt{\lambda}}+\frac{\gamma}{8\pi^{2}},3l_{\rho,U}\Big]&, \qquad
\mathcal{D}_{4}(\lambda,t) =\bigg\{ (x,y)\in \big[l_{\rho,U},2l_{\rho,U}\big]^{2} ,\, y\geq x\geq \Big(\frac{\pi}{\sqrt{\lambda}}+\frac{\sigma(t)}{2y^{3}}\Big) \bigg\}.
\end{align*}
and $ \gamma $ (resp. $ \sigma(t) $) is given in Proposition \ref{Esolo} (resp. Proposition \ref{Epair}).

Set $ \delta\in (0,1) $. There exists $ \rho_{\delta}>0 $ such that for every $ \rho\in (0,\rho_{\delta}) $ there is a Fermi energy level $ \lambda_{\rho} $, depending only on $ \rho $ and $ U $, such that, with probability $ 1-O(L^{-\infty}) $, the thermodynamic limit of the ground state energy per particle satisfies
\begin{equation}\label{asymp1}
\mathcal{E}^{U}(\rho):=\lim_{\substack{L\to +\infty\\ \frac{n}{L}\to \rho}}\frac{E_{w}^{U}(\Lambda,n)}{n}=\frac{1}{\rho} \mathcal{J}(\lambda_{\rho})+O(\rho^{2-\delta}).
\end{equation}
\end{theorem}

We also get results on the ground state itself. Recall that, in any chain of $ \mathcal{P}_{2} $, there is at most two particles. They are either in the same piece (see the operator (\ref{ope1})) either in two distinct pieces (see the operator (\ref{ope2})). From $ \lambda_{\rho}>0 $ a Fermi energy level given by Theorem \ref{mainth}, we build an occupation $ Q^{\text{test}} $ such that
\begin{enumerate}
\item[(i)] for a single piece $\Delta_{i}\in \mathcal{P}_{2}$, 
 \[ q_{i}^{\text{test}}=\max\big\{ q,\,  f^{U}(\Delta_{i},q)\leq \lambda_{\rho}\big\};\]
\item[(ii)] for a pair $ (\Delta_{j},\Delta_{k})\in \mathcal{P}_{2} $, assuming $ \vert \Delta_{j} \vert\leq \vert \Delta_{k} \vert $,
\begin{align*}
q_{j}^{\text{test}}=\max&\bigg(0,\,  \max\Big\{q, \, f^{U}\big((\Delta_{j},\Delta_{k}),q\big)\leq  \lambda_{\rho} \Big\}-1\bigg), \\
q_{k}^{\text{test}}=\min&\bigg(1,\, \max\Big\{ q,\,  f^{U}\big((\Delta_{j},\Delta_{k}),q\big)\leq \lambda_{\rho}\Big\}\bigg).
\end{align*}
\end{enumerate}
We prove that one can complete $ Q^{\text{test}} $ on $ ^{c}\mathcal{P}_{2} $ with respect to Lemma \ref{lrhoU}. Then, set the following state
\begin{equation}
\Psi^{\text{test}}(\Lambda,n)=\bigg(\bigwedge_{I\, \in \,  \mathcal{P}_{2}}\psi^{U}\Big(I,(q_{i}^{\text{test}})_{i\in I}\Big)\bigg)\wedge \bigg(\bigwedge_{I\, \in \, ^{c}\mathcal{P}_{2}}\bigwedge_{i\in I} \psi^{0}\Big(\Delta_{i},q_{i}^{\text{test}}\Big)\bigg)
\end{equation}
where
\begin{enumerate}
\item[(i)] $ \psi^{U}\Big(I,(q_{i})_{i\in I}\Big) $  is the ground state for the interacting system with exactly $ q_{i} $ particles in $ \Delta_{i} $;
\item[(ii)] $ \psi^{0}(\Delta,q) $ is the ground state for $ q $ non-interacting particles in $ \Delta $, given by the Slater determinant of the $ q $ firsts eigenfunctions of the operator  $ h_{\Delta}=-\frac{d^{2}}{dx^{2}}^{D}_{\vert \Delta} $.
\end{enumerate}
We compare the state $ \Psi^{\text{test}}(\Lambda,n) $ to any ground state $ \Psi^{U}(\Lambda,n) $ through the one- and two- particle densities, using trace norm $ \Vert \, \, \Vert_{\text{tr}} $.
\begin{definition}
For $ \phi\in \mathfrak{H}^{n}(\Lambda)  $, its \textbf{$ 1 $-particle density}  is the operator $ \gamma^{(1)}_{\phi} $ on $ \mathfrak{H}^{1}(\Lambda)=L^{2}(\Lambda) $ with kernel
\begin{equation}\label{defgamma1} 
\gamma^{(1)}_{\phi}(x,y)=n\int_{\Lambda^{n-1}} \phi(x,Z)\phi(y,Z)dZ.
\end{equation}
The \textbf{$ 2 $-particle density} of $ \phi $ is the operator $ \gamma^{(2)}_{\phi} $ on $ \mathfrak{H}^{2}(\Lambda) $ with kernel
\begin{equation}\label{defgamma2}
\gamma^{(2)}_{\phi}(x_{1},x_{2},y_{1},y_{2})=\frac{n(n-1)}{2}\int_{\Lambda^{n-2}}\phi(x_{1},x_{2},Z)\phi(y_{1},y_{2},Z)dZ.
\end{equation}
\end{definition}

\begin{proposition}\label{compgamma1}
Let $ \Psi^{U}(\Lambda,n) $ be a ground state of $ H^{U}(\Lambda,n) $. For $ \delta\in (0,1) $, $ \rho\in (0,\rho_{\delta})$, set the state $ \Psi^{\text{test}}(\Lambda,n) $ according to the above construction. Then, in the thermodynamic limit, with probability $ 1-O(L^{-\infty}) $, one has
\begin{equation}
\frac{1}{n}\Big\Vert \gamma^{(1)}_{\Psi^{U}(\Lambda,n)}-\gamma^{(1)}_{\Psi^{\text{test}}(\Lambda,n)} \Big\Vert_{\text{tr}} \leq 10\rho^{2-\delta}.
\end{equation}
\end{proposition}
We get an analogue of Proposition \ref{compgamma1} for the  $ 2 $-particle density.
\begin{proposition}\label{compgamma2}
Let $ \Psi^{U}(\Lambda,n) $ be a ground state of $ H^{U}(\Lambda,n) $. For $ \delta\in (0,1) $ and $ \rho\in (0,\rho_{\delta}) $, set the state $ \Psi^{\text{test}}(\Lambda,n) $ as above.  Then, in the thermodynamic limit, with probability $ 1-O(L^{-\infty}) $, one has
\begin{equation}
\frac{1}{n^{2}}\Big\Vert \gamma^{(2)}_{ \Psi^{U}(\Lambda,n)}-\gamma^{(2)}_{\Psi^{\text{test}}(\Lambda,n)}\Big\Vert_{\text{tr}}\leq 45\rho^{2-\delta}.
\end{equation}
\end{proposition}

\begin{remark}
Proposition \ref{compgamma1} and Proposition \ref{compgamma2} show that the state $ \Psi^{\text{test}} $ is a better approximation of the ground state than the approximated state given in \cite{Klopp2020}.
\end{remark}

$ \\ $

\section{Expressing the ground state energy for a fixed occupation}\label{decompochain}

\subsection{Proof of Lemma \ref{lrhoU}}

 Define the \textit{Fermi length} $ l_{\rho} $  as the length of a piece $ \Delta $ for which the ground state energy of the Dirichlet Laplacian $ -\frac{d^{2}}{dx^{2}}^{D}_{\vert\Delta} $ is equal to the Fermi energy. Using formula (\ref{IDS}), one computes
\begin{equation}\label{lrho}
l_{\rho}:=\frac{\pi}{\sqrt{E_{\rho}}}=-\log\Big(\frac{\rho}{1+\rho}\Big)
\end{equation}

For $ L $ large enough, with probability $ 1-O(L^{-\infty}) $ no piece of a length below $ kl_{\rho} $ can carry more than $ k-1 $ particles in the ground state of the free operator $ H^{0}(\Lambda,n) $. Due to Assumption \ref{finiterange} of finite-range interactions, in the case of the full operator $ H^{U}(\Lambda,n) $, we exhibit the same phenomenon for some minimal length $ l_{\rho,U}<l_{\rho} $. The following lemma is a reformulation of Lemma \ref{lrhoU}.
\begin{lemma}\label{lmin}
Let $ \Psi^{U}(\Lambda,n) $ to be a ground state of $ H^{U}(\Lambda,n) $. For $ n $ and $ L $ large enough, with probability $ 1-O(L^{-\infty}) $, there exists a minimal length $ l_{\rho,U}=l_{\rho}-(4M+6)\rho $ such that
\begin{center}
\textit{If a piece $ \Delta_{i} $ satisfies $ \vert \Delta_{i} \vert<kl_{\rho,U} $, $ k\in \N $, then, for every occupation $ Q $,}
$$ \Big(P_{Q}\Psi^{U}(\Lambda,n)\neq 0 \Big) \Rightarrow \Big(q_{i}\leq k-1 \Big) $$
\end{center}
where $ P_{Q} $ is the orthogonal projector on $ \mathfrak{H}_{Q} $.

Then, any ground state of $ H^{U}(\Lambda,n) $ belongs to $ \bigoplus_{Q\in \mathfrak{Q}}\mathfrak{H}_{Q}(\Lambda) $ where $ \mathfrak{H}_{Q}(\Lambda) $ is given in (\ref{decomposition}) and
\begin{equation}\label{subsetQ}
 \mathfrak{Q}=\Big\{ (q_{i})\in \N^{m},\, \sum_{i=1}^{m}q_{i}=n \text{ and for } 1\leq i\leq m \quad  q_{i}\leq \Big\lfloor\frac{l_{i}}{l_{\rho,U}}\Big\rfloor  \Big\}.
\end{equation}
\end{lemma}
This is a slight improvement of Lemma $ 3.25 $ of \cite{Klopp2020}. We use the same method of proof.
\begin{proof}
Set $ l_{\rho,U}=l_{\rho}-t\rho $, for $ t>0 $.
Assume that $ \Delta^{e} $ is the smallest piece that does not satisfy the property of the lemma. Pick $ k \in \N $ so that $ (k-1)l_{\rho,U}\leq \vert\Delta^{e}\vert < kl_{\rho,U} $ and $ Q^{e} $ an occupation so that $ \Delta^{e} $ is occupied by $ j=k-1+e $ particles in $ P_{Q^{e}}\Psi^{U}(\Lambda,n) $ with $ e\geq 1 $. Without loss of generality, we assume that  $ \Psi^{U}(\Lambda,n)=P_{Q^{e}}\Psi^{U}(\Lambda,n) $.

We show that one can define a state $ \Phi^{U}(\Lambda,n) $  such that
\[ \langle \Phi^{U}(\Lambda,n) ,H^{U}(\Lambda,n) \Phi^{U}(\Lambda,n) \rangle < \langle \Psi^{U}(\Lambda,n) ,H^{U}(\Lambda,n)\Psi^{U}(\Lambda,n)\rangle \]
by moving the $ e $ extra particles in $ e $ empty pieces without creating any interaction.

By hypothesis, there are at most $ n-j+1 $ pieces with some particle in the state $ \Psi^{U}(\Lambda,n) $.
We call \textit{interaction range} of a piece $ \Delta $ the set of pieces $ \Delta' $ such as the distance between $ \Delta $ and $ \Delta' $ is less than or equal to $ M $.
 Thanks to Proposition \ref{nb1} and Proposition \ref{nb2}, one knows, with probability $ 1-O(L^{-\infty}) $,
\begin{align*}
\#\big\{ \Delta, \,l_{\rho,U}< \vert \Delta \vert < 2 l_{\rho,U} \big\}&=Le^{-l_{\rho,U}}(1-e^{-l_{\rho,U}})+O(L^{\beta})\\
&=n(1+(t-1)\rho+o(\rho))(1-\rho+o(\rho)) \\
\#\big\{ (\Delta,\Delta'), \, \vert \Delta \vert > l_{\rho,U} ,\, \vert \Delta' \vert >l_{\rho,U}, \, \textbf{d}(\Delta,\Delta')\leq 2M+1 \big\}&=2(2M+1)Le^{-2l_{\rho,U}}+O(L^{\beta}) \\
&=2(2M+1)n(\rho+o(\rho))\big(1+2t\rho+o(\rho))\big).
\end{align*}
Thus, there are more than $ n\Big(1+(t-1)\rho - 2(2M+2)\rho +o(\rho) \Big) $ pieces of length between $ l_{\rho,U} $ and $ 2l_{\rho,U} $ such that there is no other piece of length greater than $ l_{\rho,U} $ in any interaction range and, for any two interaction ranges, their intersection is empty. This last property means that no particle can interact with some particles of both pieces.

Choose $ t=4M+6 $ so that $ n\Big(1+(t-1)\rho - 2(2M+2)\rho +o(\rho) \Big) \geq n+1 $  for $ n $ large enough.  By the pigeonhole principle, there are at least $ j $ of such pieces for which the interaction area do not carry any particle in $ \Psi^{U}(\Lambda,n) $. Therefore one can move the $ e $ extra particles to these slots. We get a new state $ \Phi^{U}(\Lambda,n) $.

 Before the exchange, the free energy of the piece $ \Delta_{e} $ is 
\begin{align*}
 E^{0}(\Delta^{e},j)&=E^{0}(\Delta^{e},k-1)+ \sum_{i=k}^{j}\frac{i^{2}\pi^{2}}{\vert \Delta^{e}\vert^{2}} \\
 &=E^{0}(\Delta^{e},k-1)+\frac{6ek^{2}+6e(e-1)k+(2e-1)e(e-1)}{6}\frac{\pi^{2}}{\vert \Delta^{e}\vert^{2}} \\
 &\geq E^{0}(\Delta^{e},k-1)+\frac{6ek^{2}+6e(e-1)k+(2e-1)e(e-1)}{6}\frac{\pi^{2}}{k^{2}l_{\rho,U}^{2}}\\
 &\geq E^{0}(\Delta^{e},k-1)+e\frac{j}{k}\frac{\pi^{2}}{l_{\rho,U}^{2}} . 
\end{align*}
So, the $ e $ extra particles contribute to more than $ e\frac{j}{k}\frac{\pi^{2}}{l_{\rho,U}^{2}} $ in $ \Psi^{U}(\Lambda,n) $. But in  $ \Phi^{U}(\Lambda,n) $, the free energy associated to these $ e $ particles is strictly less than $ e\frac{\pi^{2}}{l_{\rho,U}^{2}} $ and there is no interaction energy. So,
\[ \langle \Phi^{U}(\Lambda,n) ,H^{U}(\Lambda,n) \Phi^{U}(\Lambda,n) \rangle < \langle \Psi^{U}(\Lambda,n) ,H^{U}(\Lambda,n)\Psi^{U}(\Lambda,n)\rangle \]
Thus $ \Psi^{U}(\Lambda,n) $ can not be a ground state and this completes the proof of Lemma \ref{lmin}.
\end{proof}

\subsection{Decomposition of $ \Lambda $ into non-interacting groups of pieces}\label{decompositiondelambda}

 From now on, we fix the minimal length $ l_{\rho,U}=l_{\rho}-(4M+6)\rho $. According to Lemma \ref{lmin}, the pieces of length $ l<l_{\rho,U} $ are empty for any ground state. We divide the others pieces into undecomposable groups of pieces that may interact through $ U $. For simplicity, we identify a piece $ \Delta_{k} $ and its index $ k $ (position). The length of the piece $ k $ is denoted by $ l_{k} $ and the distance between the pieces $ j $ and $ k $ by $ \textbf{d}_{j,k} $.

\begin{definition}\label{chain}
The r-tuple $ I=(i_{1},\dots, i_{r}) $, with $ i_{1}<\dots <i_{r}  $, is a \textbf{chain} of size $ r $ if
\begin{enumerate}
\item[(i)] for every $ k\in \llbracket 1, r\rrbracket $, $ l_{i_{k}}\geq l_{\rho,U} $,
\item[(ii)] for every $ k \in\llbracket 1, r-1\rrbracket $,  $\, \textbf{d}_{i_{k},i_{k+1}}\leq M $,
\item[(iii)] for every $ j < i_{1} $ such that $ l_{j}\geq l_{\rho,U} $,
$ \textbf{d}_{j,i_{1}}>M $
\item[(iv)] for every $ j>i_{r} $ such that $ l_{j}\geq l_{\rho,U} $, $ \textbf{d}_{i_{r},j}>M $.
\end{enumerate}
\end{definition}

Fix $ p\in \N^{\star} $. We denote by
\begin{equation}
\mathcal{P}_{p}=\Big\{ I \text{ chain},\, \sum_{i\in I}\Big\lfloor \frac{l_{i}}{l_{\rho,U}}\Big\rfloor < (p+1)  \Big\}
\end{equation}
the set of chains that cannot carry more than $ p $ particles in any ground state of $ H^{U}(\Lambda,n) $, and by $ \mathcal{N}_{p} $ the set of others pieces.
Using the notations of Lemma \ref{lmin}, we consider, for a fixed occupation $ Q\in \mathfrak{Q}  $, the operator
\begin{equation}
H^{U}(\Lambda,n,Q)=P_{Q}H^{U}(\Lambda,n)P_{Q} \qquad \text{ on}  \quad \mathfrak{H}_{Q}(\Lambda)=\bigwedge_{i=1}^{m}\bigg(\bigwedge_{j=1}^{q_{i}}L^{2}(\Delta_{i})\bigg).
\end{equation}
As chains do not interact one with another, $ H^{U}(\Lambda,n,Q) $ can be written as a sum of operators each of which acting on a specific chain. We list the notations and definitions for these operators.
\begin{definition}\label{levelsolo}
Fix $ I $ a chain in $ \Lambda $. For $ (q_{i})_{i\in I}\in \N^{\star} $, let $ \psi^{U}(I,(q_{i})_{i\in I}) $ and $ E^{U}(I,(q_{i})_{i\in I}) $ be the ground state and the ground state energy of the operator $ H^{U}(I,(q_{i})_{i\in I}) $ given by
\begin{equation}\label{HUI}
H^{U}\big(I,(q_{i})_{i\in I}\big)=\sum_{\kappa=1}^{\kappa_{I}}\bigg(\bigotimes_{j=1}^{\kappa-1}\textbf{1}_{\mathfrak{H}(\Lambda)}\bigg) \otimes h_{I} \otimes \bigg(\bigotimes_{j=\kappa+1}^{\kappa_{I}} \textbf{1}_{\mathfrak{H}(\Lambda)}\bigg)+W_{\kappa_{I}} \qquad \text{ on } \quad \bigwedge_{i\in I}\Big(\bigwedge_{j=1}^{q_{i}}L^{2}(\Delta_{i})\Big)
\end{equation}
where
\begin{enumerate}
\item[(i)] $ \kappa_{I}=\sum_{i\in I}q_{i} $ is the number of particles in $ I $;
\item[(ii)] $ h_{I} $ is the one-particle operator defined by \begin{equation} 
h_{I}=\bigoplus_{i\in I}\bigg(-\frac{d^{2}}{dx^{2}}^{D}_{\vert \Delta_{i}}\bigg) \qquad \text{ on } \mathfrak{H}(\Lambda);
\end{equation}
\item[(iii)] $ W_{k} $ is given by (\ref{W}).
\end{enumerate}
Set $ F^{U}(I,0)\equiv 0 $ and for $ \kappa \in \N^{\star} $
\begin{equation}
F^{U}(I,\kappa)=\min_{\kappa_{I}=k}E^{U}(I,(q_{i})_{i\in I}).
\end{equation}
For $ \kappa\in \N^{\star} $, the \textbf{$ \kappa $-th energy level of the chain $ I $} is defined by
\begin{equation}
f^{U}(I,\kappa)=F^{U}(I,\kappa)-F^{U}(I,\kappa-1).
\end{equation}
\end{definition}
With the notations of Definition \ref{levelsolo}, $ \psi^{U}(\Lambda,n,Q) $ the ground state of $ H^{U}(\Lambda,n,Q) $ has the form
\begin{equation}\label{statefull}
\psi^{U}(\Lambda,n,Q)=\psi^{U}_{\mathcal{P}_{p}}(Q)\wedge \psi^{U}_{\mathcal{N}_{p}}(Q)
\end{equation}
where
\begin{equation}\label{statedecompchain}
\psi^{U}_{\mathcal{P}_{p}}(Q)=\bigwedge_{I\in \mathcal{P}_{p}}\psi^{U}(I,(q_{i})_{i\in I}) \qquad \text{ and } \qquad \psi^{U}_{\mathcal{N}_{p}}(Q)=\bigwedge_{I\text{ chain } \subset\,  \mathcal{N}_{p}}\psi^{U}(I,(q_{i})_{i\in I})
\end{equation}

The corresponding ground state energy is
\begin{equation}\label{energystatefull}
E^{U}(\Lambda,n,Q)=E^{U}_{\mathcal{P}_{p}}(Q)+E^{U}_{\mathcal{N}_{p}}(Q)
\end{equation}
with
\begin{equation}
E^{U}_{\mathcal{P}_{p}}(Q)=\sum_{I \in \mathcal{P}_{p}}E^{U}(I,(q_{i})_{i\in I}) \quad \text{ and } \quad  E^{U}_{\mathcal{N}_{p}}(Q)=\sum_{I \text{ chain } \subset \, \mathcal{N}_{p}}E^{U}(I,(q_{i})_{i\in I}).
\end{equation}
We study these two quantities in the next subsections.

\subsection{Study of $ E^{U}_{\mathcal{N}_{p}} $}

The following lemma give an upper bound for the number of particles that one does not control when the occupation is known only for the chains of $ \mathcal{P}_{p} $.
\begin{lemma}\label{occupoutofchains}
For $ p\in \N^{\star} $, and $ \delta\in (0,1) $, there exists $ \rho_{\delta}>0 $ such that for every $ \rho\in (0,\rho_{\delta}) $
\begin{equation}\label{encadrement}
 \rho^{\, p+\delta}\leq\sup_{(q_{i})\in \mathfrak{Q}}\bigg(\frac{1}{n}\sum_{i\in \mathcal{N}_{p}}q_{i}\bigg)\leq \rho^{\, p-\delta}.
\end{equation}
\end{lemma}

\begin{proof}
 If $ i\in \mathcal{N}_{p} $, we have the following options.
\begin{enumerate}
\item[(i)] Either $ l_{i}<l_{\rho,U} $, $ q_{i}=0 $;
\item[(ii)] Or $ l_{i}\geq (p+1)l_{\rho,U} $, then, using Proposition \ref{nb1}, one computes
\[ \sum_{i,\, l_{i}\geq (p+1)l_{\rho,U} }q_{i}\leq \sum_{k=p+1}^{+\infty}kL(e^{-kl_{\rho,U}}-e^{-(k+1)l_{\rho,U}})= (p+1)Le^{-(p+1)l_{\rho,U}}(1+O(e^{-l_{\rho,U}}))  \]
\item[(iii)] Or $ i\in I $ chain of size $ r\geq 2 $ and $ \sum_{j\in I}l_{j}\geq (p+1)l_{\rho,U} $ and $ l_{i}<(p+1)l_{\rho,U} $; in this case $ q_{i}\leq p $. For $ r\leq p $,
\begin{align*}
\# \{I \text{ chain of size } r \text{ of total length } \geq (p+1)l_{\rho,U}\} &\leq \# \{ r \text{ pieces of total length } \geq (p+1)l_{\rho,U} \\
& \quad \qquad \text{ with gaps of length } \leq M \}\\
&\leq M^{r-1}Le^{-(p+1)l_{\rho,U}}
\end{align*} 
and
\begin{align*}
\# \{I \text{ chain of size } r\geq p+1\} &\leq \# \{ (p+1) \text{ pieces of length } \geq l_{\rho,U} \text{ with gaps of length } \leq M \}\\
&\leq M^{p}Le^{-(p+1)l_{\rho,U}}.
\end{align*}
\end{enumerate}
As $ e^{-(p+1)l_{\rho,U}}=o(\rho^{\, p+1-\delta}) $, this completes the proof of the right-hand side of the inequality (\ref{encadrement}).

Concerning the left-hand side, let $ Q^{0}=(q_{i}^{0}) $ be the occupation of the ground state for the free model. We have that for $ i\in \llbracket 1,m \rrbracket $ if $ l_{i}\in [kl_{\rho},(k+1)l_{\rho}) $ then $ q_{i}^{0}=k $. Since $ l_{\rho,U}\leq l_{\rho} $, $ Q^{0}\in \mathfrak{Q} $. So,
\[ \sum_{i\in \mathcal{N}}q_{i}^{0}\geq \sum_{i ,\, l_{i}\geq (p+1)l_{\rho} }q_{i}^{0}= \sum_{k=p+1}^{+\infty}kL(e^{-kl_{\rho}}-e^{-(k+1)l_{\rho}})= (p+1)Le^{-(p+1)l_{\rho}}(1+O(e^{-l_{\rho}})) \]
As $ \rho^{\, p+1+\delta}=o(e^{-(p+1)l_{\rho}}) $, it gives the left part of the inequality (\ref{encadrement}).
\end{proof}

\begin{proposition}\label{leftover}
For a fixed $ p\geq 1 $, $ \delta\in (0,1) $ and $ Q \in \mathfrak{Q} $, there exists $ \rho_{\delta}>0 $ such that for $ \rho\in (0,\rho_{\delta}) $,
\begin{equation}
E^{U}_{\mathcal{N}_{p}}(Q)\leq n\rho^{p-\delta}
\end{equation}
\end{proposition}

\begin{proof}

As in Definition \ref{levelsolo}, for any chain $ I $, we denote $ \psi^{U}\big(I,(q_{i})\big) $ and $ E^{U}\big(I,(q_{i})\big) $ the ground state and ground state energy of the operator $ H^{U}\big(I,(q_{i})\big) $ given by (\ref{HUI}). We use the notations $ \psi^{0}\big(I,(q_{i})\big) $ and $ E^{0}\big(I,(q_{i})\big) $ for the free case. We have
\[\big\langle\, \psi^{U}\big(I,(q_{i})\big)\, , \,  H^{U}\big(I,(q_{i})\big) \psi^{U}\big(I,(q_{i})\big) \, \big\rangle \leq \big\langle \,  \psi^{0}\big(I,(q_{i})\big)\, , \, H^{U}\big(I,(q_{i})\big)\psi^{0}\big(I,(q_{i})\big) \, \big\rangle \]
so
\begin{equation}
E^{U}\big(I,(q_{i})\big) \leq E^{0}\big(I,(q_{i})\big) + \big\langle \psi^{0}\big(I,(q_{i})\big),W_{\kappa_{I}}\psi^{0}\big(I,(q_{i})\big) \big\rangle.
\end{equation}
Then, we compute
\begin{align}\label{psibadbeta1}
E^{U}_{\mathcal{N}_{p}}(Q) &= \sum_{I\subset \, \mathcal{N}_{p} \text{ chain}}E^{U}\big(I,(q_{j})_{j\in I}\big) \\
&\leq\sum_{I\subset \, \mathcal{N}_{p} \text{ chain}}\bigg(E^{0}\big(I,(q_{j})_{j\in I}\big)+\big\langle \psi^{0}\big(I,(q_{i})\big),W_{\kappa_{I}}\psi^{0}\big(I,(q_{i})\big) \big\rangle \bigg) \nonumber \\
&\leq \max_{\substack{Q\in \mathfrak{Q}\\j\in \mathcal{N}_{p}}}\bigg(\frac{E^{0}(l_{j},q_{j})}{q_{j}}\bigg)\sum_{j\in \mathcal{N}_{p}}q_{j}+ \sum_{I\subset \, \mathcal{N}_{p} \text{ chain} }\big\langle \psi^{0}\big(I,(q_{i})\big),W_{\kappa_{I}}\psi^{0}\big(I,(q_{i})\big) \big\rangle  \nonumber
\end{align}
For any $ Q\in \mathfrak{Q} $ and $ j\in \mathcal{N}_{p} $, by Lemma \ref{lmin},
\begin{equation}\label{bornelibre}
E^{0}(l_{j},q_{j})=\sum_{k=1}^{q_{j}}\frac{k^{2}\pi^{2}}{l_{j}^{2}}
\leq C\frac{q_{j}^{3}}{l_{j}^{2}}
\leq C\frac{q_{j}}{l_{\rho,U}^{2}}.
\end{equation}
By Lemma \ref{occupoutofchains}, $ \sum_{j\in \mathcal{N}_{p}}q_{j}\leq n\rho^{p-\delta} $. 

We deal with the remaining sum using the results of Lemma \ref{calculinteraction}. For a chain $ I $, $ i\in I $, $ j\in \N $, let $ \phi^{\Delta_{i}}_{j} $ be the state on $ L^{2}(\Delta_{i}) $ given by
\begin{equation}
\phi^{\Delta_{i}}_{j}(x)=\frac{\sqrt{2}}{\sqrt{l_{i}}}\sin\Big(\frac{\pi}{l_{i}} j(x-x_{i})\Big)\mathbf{1}_{\Delta_{i}}(x).
\end{equation}
Then,
\begin{equation}
\psi^{0}(I,(q_{i}))=\bigwedge_{i\in I}\bigwedge_{j=1}^{q_{i}}\phi^{\Delta_{i}}_{j}.
\end{equation}
By skew-symmetry and orthogonality of $ (\phi^{\Delta_{i}}_{j})_{i,j} $,
\begin{align*}
\big\langle \psi^{0}\big(I,(q_{i})\big),W_{\kappa_{I}}\psi^{0}\big(I,(q_{i})\big) \big\rangle&= \frac{\kappa_{I}(\kappa_{I}-1)}{2}\int U(y-x)\psi^{0}\big(I,(q_{i})\big)^{2}(y,x,Z) \, dxdydZ\\
&=\sum_{i\in I}\sum_{1\leq j<k\leq q_{i}}\int U(y-x)\Big\vert \phi^{\Delta_{i}}_{j}\wedge \phi^{\Delta_{i}}_{k}\Big\vert^{2}(x,y)dxdy \\
&\quad + \sum_{h,i \, \in I, \, h\neq i}\sum_{j=1}^{q_{h}}\sum_{k=1}^{q_{i}}\int U(y-x)\Big\vert \phi^{\Delta_{h}}_{j}\wedge \phi^{\Delta_{i}}_{k}\Big\vert^{2}(x,y)dxdy
\end{align*}
So, by Lemma \ref{calculinteraction},
\begin{align}\label{borneinter}
\big\langle \psi^{0}\big(I,(q_{i})\big),W_{p}\psi^{0}\big(I,(q_{i})\big) \big\rangle&\leq C \sum_{i\in I}\sum_{1\leq j<k\leq q_{i}}\frac{j^{2}+k^{2}}{l_{i}^{3}}+C\sum_{h,i \, \in I, \, h\neq i}\sum_{j=1}^{q_{h}}\sum_{k=1}^{q_{i}}\frac{j^{2}k^{2}}{l_{h}^{3}l_{i}^{3}} \\
&\leq C\sum_{i\in I}\frac{q_{i}}{l_{\rho,U}^{3}}+C\sum_{h,i \, \in I, \, h\neq i}\frac{q_{h}q_{i}}{l_{\rho,U}^{6}} \nonumber \\
&\leq \frac{C}{l_{\rho,U}^{3}}\sum_{i\in I}q_{i}+\frac{C}{l_{\rho,U}^{6}}\Big( \sum_{i\in I}q_{i}\Big)^{2}  \nonumber
\end{align}
where $ C $ depends on $ U $ and $ M $. Again by Lemma \ref{occupoutofchains}, $ \sum_{i\in \mathcal{N}_{p}}q_{i}\leq n\rho^{p-\varepsilon} $. For the part with squares, we adapt the proof of (\ref{encadrement}). A chain $ I\subset \mathcal{N}_{p} $ of size $ r\geq p+1 $ of total length $ l\in [kl_{\rho,U},(k+1)l_{\rho,U}) $ with $ k\geq r $ may contain at most $ k $ particles. Otherwise, the chains $ I\subset \mathcal{N} $  of size $r\leq p $ and of total length $ l\in [kl_{\rho,U},(k+1)l_{\rho,U}) $ with $ k\geq p+1 $ may contain at most $ k $ particles.
So,
\[ \sum_{I\subset \mathcal{N}_{p}}\Big(\sum_{i\in I}q_{i}\Big)^{2}\leq \sum_{r=p+1}^{+\infty}M^{r-1}\sum_{k=r}^{\infty}k^{2}Le^{-kl_{\rho,U}} +\sum_{r=1}^{p}M^{r-1}\sum_{k=p+1}^{+\infty}k^{2}Le^{-kl_{\rho,U}}. \]
We claim that, if $ Me^{-l_{\rho,U}}<1 $,
\begin{equation}\label{bornecarr}
\exists C>0 \qquad \sum_{I\subset \mathcal{N}_{p}}\Big(\sum_{i\in I}q_{i}\Big)^{2}\leq C\max\{1,\dots, M^{p}\}(p+1)^{2}Le^{-(p+1)l_{\rho,U}}.
\end{equation}

Then, if $ \rho $ is small enough, combining (\ref{bornelibre}), (\ref{borneinter}) and (\ref{bornecarr}), the inequality (\ref{psibadbeta1}) becomes
\begin{equation}\label{psibadbeta2}
E^{U}_{\mathcal{N}_{p}}(Q)\leq n\rho^{p-\delta}.
\end{equation}
It concludes the proof of Proposition \ref{leftover}.
\end{proof}

\begin{remark}\label{remark1}
If one replaces $ \psi^{U}_{\mathcal{N}_{p}}(Q) $ by $ \bigwedge_{I \subset \mathcal{N}_{p}}\psi^{0}\big(I,(q_{i})\big) $ then the same bound holds for the energy.
\end{remark}

\subsection{Study of $ E^{U}_{\mathcal{P}_{p}} $}

The following proposition states that, when the number of particles in $ \mathcal{P}_{p} $ is known, $  E^{U}_{\mathcal{P}_{p}} $ is the sum of the smallest energy levels. But it requires a strong hypothesis on the monotony of the energy levels.

\begin{assumption}\label{convex}
For a fixed $ p\geq 1 $, using the notations of Definition \ref{levelsolo}, the application
\begin{equation}
f^{U}(I,.): \begin{cases} \, \llbracket 0,p \rrbracket &\longrightarrow \R \\
\quad r \qquad &\longmapsto f^{U}(I,r) \end{cases}
\end{equation}
is increasing for every chain $ I $ in $ \mathcal{P}_{p} $.
\end{assumption}

From now on, we denote by $ n_{Q} $ the \textit{number of particles in $ \mathcal{P}_{p} $ for the occupation} $ Q $. Using the notations of Definition \ref{levelsolo}, we also set
\begin{equation}\label{Gammap}
\Gamma_{p}=\Big\{ f^{U}(I,k), \, I \in \mathcal{P}_{p}, 1\leq k\leq p \Big\}.
\end{equation}\label{orderp}
Let $ \leq_{p} $ be a lexical order on $ \Gamma_{p} $ such that
\begin{equation}
\forall I,J \in \mathcal{P}_{p}, \, 1\leq k,l\leq p \qquad f^{U}(I,k)<_{p}f^{U}(J,l) \Longleftrightarrow \begin{cases} f^{U}(I,k)<f^{U}(J,l) \\  \text{else} \quad \text{last index of }  I < \text{first index of J} \\\text{else } \quad  k<l \end{cases}
\end{equation}

\begin{proposition}\label{levelsomme}
For a fixed $ p\geq 1 $, let $ \{a_{k}\in \Gamma_{p}, \,   a_{k-1}<_{p} a_{k} \} $ be the ordered set given by (\ref{Gammap}) and (\ref{orderp}). 
Under Assumption \ref{convex}, for $ r\leq \min(n,\# \Gamma_{p}) $, any occupation $ Q $ that minimizes $ E^{U}_{\mathcal{P}_{p}} $ when $ n_{Q}=r $ and $ q_{i}\leq \Big\lfloor \frac{l_{i}}{l_{\rho,U}}\Big\rfloor $ for each piece $ \Delta_{i} $ in $ \mathcal{P}_{p} $ , satisfies
\begin{equation}\label{Emin}
E^{U}_{\mathcal{P}_{p}}(Q)=\sum_{k=1}^{r}a_{k}.
\end{equation}
\end{proposition}

\begin{proof}
Fix $ r\leq \min(n,\# \Gamma_{p})  $. Take such an occupation $ Q $. Then, by reductio ad absurdum,
\[ E^{U}_{\mathcal{P}_{p}}(Q)=\sum_{I\in \mathcal{P}_{p}}E^{U}(I,(q_{i})_{i\in I})=\sum_{I\in \mathcal{P}_{p}}F^{U}(I,\kappa_{I})=\sum_{I\in \mathcal{P}_{p}}\sum_{j=1}^{\kappa_{I}}f^{U}(I,j) \]
with $ \sum_{I\in \mathcal{P}_{p}}\kappa_{I}=r $ and $ \kappa_{I}\leq \sum_{i\in I}\Big\lfloor \frac{l_{i}}{l_{\rho,U}}\Big\rfloor \leq p  $.
In particular,
\begin{equation}
\sum_{k=1}^{r}a_{k}\leq E^{U}_{ \mathcal{P}_{p}}(Q).
\end{equation}

For the reverse inequality, we build by induction an appropriate occupation $ Q' $.  Set $ Q'(0)\equiv (0,\dots, 0) $. For $ k $ from $ 1 $ to $ r $, assume that the multi-index $ Q'(k-1)=\big(q_{i}'(k-1)\big)_{1\leq i\leq m} $ satisfies
$$ E^{U}_{\mathcal{P}_{p}}\big(Q'(k-1)\big)=\sum_{s=1}^{k-1}a_{s} \qquad \text{and} \qquad \sum_{i=1}^{m}q_{i}'(k-1)=k-1. $$
We know that $ a_{k}=f^{U}(I,j) $ meaning $ a_{k} $ is the $ j $-th energy level of the chain $ I $. Since $ f^{U}(I,.) $ is increasing, we have $ \{f^{U}(I,1),\dots ,f^{U}(I,j-1)\}= \{a_{i_{1}},\dots ,a_{i_{j-1}} \} $ for $ 1\leq i_{1}<\dots<i_{j-1}\leq k-1 $ and for every $ i>j $, $ f^{U}(I,i)>a_{k} $. In particular,
$$ \sum_{s=1}^{k}a_{s}=\sum_{s\notin \{i_{1},\dots, i_{j-1},k\}}a_{s}+\sum_{i=1}^{j}f^{U}(I,i)=\sum_{s\notin \{i_{1},\dots, i_{j-1},k\}}a_{s}+F^{U}(I,j). $$
We set $ q_{i}'(k) $ for $ i\in I $ so that $ E^{U}(I,(q_{i}'(k))_{i\in I})=F^{U}(I,j) $ and for every $ i\notin I $, $ q_{i}'(k)=q_{i}'(k-1). $ Then,
$$ E^{U}_{\mathcal{P}_{p}}\big(Q'(k)\big)=\sum_{s=1}^{k}a_{s} \qquad \text{and} \qquad \sum_{i=1}^{m}q_{i}'(k)=k. $$
We fill the coordinates in $ \mathcal{N}_{p} $ so that $ Q' $ is an occupation with $ n_{Q}=r $.

It concludes the proof of Proposition \ref{levelsomme}.
\end{proof}

\paragraph{\textbf{Remark}} We don't know yet how to prove that Assumption \ref{convex} holds when $ p\geq 3 $. The following lemma gives a partial result for chains of size $ 1 $.

\begin{lemma}\label{convexsolo}
Using the notations of Definition \ref{levelsolo}, if $ l<l_{\rho,U}^{\frac{3}{2}-\varepsilon} $ for $ \varepsilon\in (0,\frac{1}{2}) $ then for $ \rho $ small enough
\begin{equation}
\forall r\in \llbracket 1,p-1\rrbracket \qquad f^{U}\big([0,l],r\big)<f^{U}\big([0,l],r+1\big)
\end{equation}  where $ p=\lfloor \frac{l}{l_{\rho,U}} \rfloor $.
\end{lemma}
\begin{proof}
 Assume that $ l =\pi\beta^{-1}l_{\rho,U} $.
For $ r\geq 1 $, denote $ F(r)=F^{U}\big([0,l], r\big) $ and $ \Psi (r) $ a corresponding eigenfuntion. Also set $ F^{0}(r) $ and $ \Psi^{0}(r) $ in the free case $ U\equiv 0 $. We know that
$$ \Psi^{0}(r)=\bigwedge_{i=1}^{r}\varphi_{i} $$
where $ \varphi_{i}(x)=\frac{\sqrt{2}}{\sqrt{l}}\sin\big(\frac{\pi}{l} ix\big)\mathbf{1}_{[0,l]}(x) $.
We compute
\begin{align*}
\langle W_{r}\Psi^{0}(r),\Psi^{0}(r)\rangle &= \frac{r(r-1)}{2}\int U(x_{1}-x_{2})\Psi^{0}(r)(x_{1},x_{2},Z)^{2}dx_{1}dx_{2}dZ \text{ by skew-symmetry}\\
&=\frac{r(r-1)}{2}\frac{1}{r!}\sum_{\sigma, \sigma ' \in \mathfrak{S}_{r}}\varepsilon(\sigma)\varepsilon(\sigma') \int U(x_{1}-x_{2}) \prod_{i=1}^{r}\varphi_{\sigma(i)}(x_{i})\varphi_{\sigma'(i)}(x_{i})dX \\
&= \sum_{p<q\leq r}\int U(x_{1}-x_{2})\big\vert\varphi_{p}\wedge\varphi_{q}\big\vert^{2}(x_{1},x_{2})dx_{1}dx_{2}
\end{align*}
by skew-symmetry and orthogonality of $(\varphi_{i})_{i\geq 1}$. Hence, by Lemma \ref{calculinteraction},
\begin{equation}
\langle W_{r}\Psi^{0}(r),\Psi^{0}(r)\rangle \leq \sum_{p<q\leq r}Cl^{-3}(p^2+q^2) \leq
Cl^{-3}r^{4}.
\end{equation}
Since,
$$ 0\leq F(r)-F^{0}(r)\leq  \langle W_{r}\Psi^{0}(r),\Psi^{0}(r)\rangle. $$
we have
\begin{equation}
F(r)=F^{0}(r)+O(l^{-3}r^{4})=\sum_{i=1}^{r}(\pi l^{-1}  i)^{2}+O(l^{-3}r^{4})
\end{equation}
Then,
\begin{equation}\label{convex1}
F(r+1)-2F(r)+F(r-1)=2\pi l^{-2}r\big(1+O(l_{\rho,U}^{-2\varepsilon})\big)
\end{equation}
as $ r<l.l_{\rho,U}^{-1} $ and $ l^{2}l_{\rho,U}^{-3}\leq l_{\rho,U}^{-2\varepsilon} $. Thus one gets that for $ \rho $ small enough the r.h.s is positive. This concludes the proof of Lemma \ref{convexsolo}.
\end{proof}
Combining Lemma \ref{convexsolo} and Lemma \ref{convexsomme}, we get that Assumption \ref{convex} holds when one cancels the interaction between pieces and $ p $ is less than $ \vert\log(\rho) \vert $. Without restriction on the form of the interaction, the issue occurs when the growth in the free energy is less or of the order of the interaction between two pieces. More precisely, we don't know yet how to deal with the cases where the lengths of a pair of pieces $ \{ \Delta_{i},\Delta_{j}\} $ satisfy
\begin{equation}
 \exists \, k_{i},k_{j}\in \llbracket 1, p-1\rrbracket  \qquad \Big\vert\Big(\frac{k_{i}}{l_{i}}\Big)^{2}-\Big(\frac{k_{j}}{l_{j}}\Big)^{2}\Big\vert = O(l_{\rho,U}^{-6}).
 \end{equation}
 
$ \\ $

Otherwise, let $ (\Gamma_{p},\leq_{p}) $ be the ordered set given by (\ref{Gammap}) and (\ref{orderp}), and for $ 1\leq r \leq \#\Gamma_{p} $, and let $ \mathcal{G}_{p}(r) $ be the subset of $ \mathcal{P}_{p} $ such that
\begin{equation}\label{equal}
\mathcal{G}_{p}(r)=\Big\{ I \in \mathcal{P}_{p}, \, \exists \,  1\leq k\leq p ,\, f^{U}(I,k)=\text{the } r\text{-th smallest element of }  (\Gamma_{p},\leq_{p}) \Big\}.
\end{equation}

From the proof of Proposition \ref{levelsomme} we deduce the following corollary.
For $ B\subset \llbracket 1,m\rrbracket $ and $ Q\in \N^{m} $, we denote by $ Q_{\vert B}=(q_{i})_{i\in B} $ the \textit{restriction  of the multi-index to} $ B $. 

\begin{corollary}\label{sequenceQ}
Under Assumption \ref{convex}, there exists a sequence of occupations $ \big(Q(r)\big)_{r\leq n} $ such that
\begin{enumerate}
\item the number of particles in the chains of $ \mathcal{P}_{p} $ for the occupation $ Q(r) $ is $ n_{Q(r)}= r $; 
\item the restrictions $ Q(r)_{\vert \mathcal{P}_{p}} $ and $ Q(r+1)_{\vert \mathcal{P}_{p}} $ are equal except for one chain;
\item if $ \Psi^{U} $ is a ground state of $ H^{U} $ and $ Q $ is an occupation that satisfies $ P_{Q}\Psi^{U}\neq 0 $ then 
\begin{equation}
Q_{\vert \mathcal{P}_{p}\backslash \mathcal{G}_{p}(n_{Q})}=Q(n_{Q})_{\vert \mathcal{P}_{p}\backslash \mathcal{G}_{p}(n_{Q})},
\end{equation}
where $ \mathcal{G}_{p}(n_{Q}) $ is given by (\ref{equal}).
\end{enumerate} 
\end{corollary}

The issue of the cardinal of $ \mathcal{G}_{p}(r) $, for any $ r\in \llbracket 1,\,  \#\Gamma_{p} \rrbracket  $, looks as hard to solve as the issue of order of degeneracy of the ground state of $ H^{U}(\Lambda,n) $. However, it seems relevant to assume that, except for some pathological Poisson point processes, one should get only few cases of equality for the energy levels of $ \Gamma_{p} $.

\begin{assumption}\label{casequal}
For $ 1\leq r \leq \#\Gamma_{p}  $, $ \# \mathcal{G}_{p}(r)\leq n\rho^{p-\delta} $.
\end{assumption}

The next proposition states that if Assumption \ref{convex} and Assumption \ref{casequal} are true for some $ p\geq 1 $ then the number of particles in each piece of $ \mathcal{P}_{p} $, except for at most $ 2n\rho^{p-\delta} $ chains, stays the same for any ground state.

\begin{proposition}\label{Qfix}
Set $ p\in \N^{\star} $, $ \delta \in (0,1) $ and $ \rho\in (0,\rho_{\delta}) $. Under Assumption \ref{convex} and Assumption \ref{casequal}, there exist a subset $ \mathcal{F}_{p} $ of $ \mathcal{P}_{p} $ and,  for each piece $ i $ in $ \mathcal{F}_{p} $, an integer $ q_{i}^{\mathcal{F}_{p}} $  such that
\begin{enumerate}
\item the number of chains in $ \mathcal{P}_{p}\backslash \mathcal{F}_{p} $ is less than or equal to $ 2n\rho^{p-\delta} $;
\item if $ \Psi^{U} $ is a ground state of $ H $ then it admits the decomposition $ \Psi^{U}=\Phi^{U,\mathcal{F}_{p}}\wedge \Omega^{U,\mathcal{F}_{p}^{c}} $ with
\begin{equation}
\Phi^{U,\mathcal{F}_{p}}=\bigwedge_{I\in \mathcal{F}_{p}}\psi^{U}\big(I,(q_{i}^{\mathcal{F}_{p}})_{i\in I}\big) \qquad \text{ and } \qquad \Omega^{U,\mathcal{F}_{p}^{c}}=\sum_{Q\in \mathfrak{Q}}\lambda(Q)\bigwedge_{I\notin \mathcal{F}_{p}}\psi^{U}\big(I,(q_{i})_{i\in I} \big).
\end{equation}
\end{enumerate}
\end{proposition}

\begin{proof}
Let $ \mathcal{F}_{p} $ be the set of chains $ I $ in $ \mathcal{P}_{p} $ such that the function $ r\mapsto Q(r)_{\vert I} $ is constant on $ \llbracket n-2n\rho^{p-\delta},n\rrbracket $.
By Corollary \ref{sequenceQ}, for $ r\leq n-1 $, there is a unique chain $ I\in \mathcal{P}_{p} $ for which $ Q(r)_{\mathcal{P}_{p}\backslash{I}}=Q(r+1)_{\mathcal{P}_{p}\backslash{I}} $. So, by induction on $ r\geq n-2n\rho^{p-\delta} $, $ \mathcal{F}_{p} $ is not empty and the numbers of chains in $ \mathcal{P}_{p}\backslash \mathcal{F}_{p} $ is less than or equal to $ 2n\rho^{p-\delta} $. Then, for any piece $ i $ in $ \mathcal{F}_{p} $, let $ q_{i}^{\mathcal{F}_{p}} $ be the common value.

Let $ \Psi^{U} $ be a ground state of $ H $. By Lemma \ref{lmin} and Definition \ref{levelsolo}, we have the decomposition
\begin{equation}\label{decompPsi1}
\Psi^{U}=\sum_{Q\in \mathfrak{Q}}\lambda(Q)\bigwedge_{I \text{ chain }}\psi^{U}\big(I,(q_{i})_{i\in I}\big)
\end{equation}
where $ \psi^{U}\big(I,(q_{i})_{i\in I}\big) $ is a normalized wave function of $ \mathfrak{H}^{q_{I}}(U_{I}) $ with $ q_{I}=\sum_{i\in I}q_{i} $ and $ U_{I}=\bigcup_{i\in I}\Delta_{i} $.

Using Lemma \ref{occupoutofchains}, if an occupation $ Q $ satisfies $ P_{Q}\Psi^{U}\neq 0 $ then the number of particles in $ \mathcal{P}_{p} $ for $ Q $ belongs to  $ \llbracket n-n\rho^{p-\delta}, n\rrbracket $. 

Under Assumption \ref{casequal}, we have
\begin{equation}
\bigcup_{r=n-n\rho^{p-\delta}}^{n} \mathcal{G}_{p}(r)\subset \mathcal{P}_{p}\backslash\mathcal{F}_{p} 
\end{equation}
because the left term, gathers all the chains that match with any $ r $-th smallest element in $ (\Gamma_{p},\leq_{p}) $ for $ r\in \llbracket n-n\rho^{p-\delta}, n\rrbracket $ (see (\ref{equal})) while the right term gathers all the chains that match with any element of $ (\Gamma_{p},\leq_{p}) $ between the $ (n-2\rho^{p-\delta}) $-th and the $ n $-th ones.

Using the third point of Corollary \ref{sequenceQ}, one shows that, for every chain $ I\in \mathcal{F}_{p} $, the restriction map $ Q\mapsto Q_{\vert I} $ is constant on $ \{ Q\in \mathfrak{Q}, P_{Q}\Psi^{U}\neq 0 \} $, equal to $ (q_{i}^{\mathcal{F}_{p}})_{i\in I} $. So, (\ref{decompPsi1}) becomes
\begin{equation}
\Psi^{U}=\bigg(\bigwedge_{I\in \mathcal{F}_{p}}\psi^{U}\big(I,(q_{i}^{\mathcal{F}_{p}})_{i\in I}\big)\bigg)\wedge\bigg(\sum_{Q\in \mathfrak{Q}}\lambda(Q)\bigwedge_{I\notin \mathcal{F}_{p}}\psi^{U}\big(I,(q_{i})_{i\in I} \big)\bigg)
\end{equation}

This concludes the proof of Proposition \ref{Qfix}.
\end{proof}

$ \\ $

\section{Proceeding with the case $p=2$}\label{split2}

 \subsection{Monotony of the energy levels}
We recall that $ \mathcal{P}_{2} $ is the set of chains each of which carries at most two particles for any ground state.

\begin{lemma}\label{nbUpsilon}
Set
\begin{equation}\label{Gamma2}
\Gamma_{2}=\Big\{ f^{U}(I,k), \, I \in \mathcal{P}_{2},\,  k\in \{1,2\} \Big\}
\end{equation}
Then, with probability $ 1-O(L^{-\infty}) $, for $ \rho $ small enough,
$$ 2n<\# \Gamma_{2} <2n\big(1+(3M+6)\rho\big). $$
\end{lemma}

\begin{proof}
Using $ e^{-l_{\rho,U}}=\rho\big(1+(4M+5)\rho+o(\rho)\big) $, Proposition \ref{nb1} and Proposition \ref{nb2}, we compute
\begin{align*}
\# \Gamma_{2}&=2\# \Big\{ \{ \Delta_{i}\}\in \mathcal{P}_{2}\Big\}
+2 \# \Big\{ \{\Delta_{j},\Delta_{k}\}\in \mathcal{P}_{2} \Big\}\\
&=2L(1-Me^{-l_{\rho,U}})^{2}\Big( \big(e^{-l_{\rho,U}}-e^{-3l_{\rho,U}}\big)+M \big(e^{-l_{\rho,U}}-e^{-2l_{\rho,U}}\big)^{2}\Big)\\
&=2L\big(1-2M\rho+o(\rho)\big)\rho\big(1+(4M+5)\rho+o(\rho)\big)\big(1+M\rho+o(\rho)\big) \\
&=2n\big(1+(3M+5)\rho +o(\rho)\big)
\end{align*}
It concludes the proof of Lemma \ref{nbUpsilon}.
\end{proof}
We now prove that Assumption $ \ref{convex} $ holds when $ p=2 $. 
 \begin{lemma}\label{convfor2}
For $ I\in \mathcal{P}_{2} $, $ f^{U}(I,2)>f^{U}(I,1) $.
\end{lemma}
\begin{proof}
If $ I\in \mathcal{P}_{2} $, then we have two cases.
\begin{enumerate}
\item[(i)] Either $ I= (i) $ is a unique piece of length $ l_{i}\in [l_{\rho,U},3l_{\rho,U}) $. The first energy level of $ \Delta_{i} $ is
$$ f^{U}(\Delta_{i},1)=\frac{\pi^{2}}{l_{i}^{2}} $$
For the second energy level of $ \Delta_{i} $, we use Proposition \ref{Esolo}.
\begin{equation}\label{e(l,2)}
f^{U}(\Delta_{i},2)=\frac{4\pi^{2}}{l_{i}^{2}}+\frac{\gamma}{l_{i}^{3}}+o(l^{-3})> f^{U}(\Delta_{i},1).
\end{equation}

$ \\ $

\item[(ii)] Or $ I=(j,k) $ is a pair of pieces of length $ l_{j},l_{k}\in [l_{\rho,U},2l_{\rho,U}) $ separated by a gap of length $ \textbf{d}_{jk}\leq M $. The first energy level of the pair $\{\Delta_{j},\Delta_{k}\} $ is
\begin{equation}\label{e(j,k,1)}
f^{U}\Big(\{\Delta_{j},\Delta_{k}\},1\Big)=\min \Big(\frac{\pi^{2}}{l_{j}^{2}},\frac{\pi^{2}}{l_{k}^{2}}\Big).
\end{equation}
Concerning the second energy level of this pair, we use Proposition \ref{Epair}.
\begin{equation}\label{e(j,k,2)}
f^{U}\Big(\{\Delta_{j},\Delta_{k}\},2\Big)=\max \Big(\frac{\pi^{2}}{l_{j}^{2}},\frac{\pi^{2}}{l_{k}^{2}}\Big)+\frac{\tau (\textbf{d}_{jk})}{l^{3}l'^{3}}\Big(1+o(1)\Big)>f^{U}\Big(\{\Delta_{j},\Delta_{k}\},1\Big).
\end{equation}
\end{enumerate}
This completes the proof of Lemma \ref{convfor2}.
\end{proof}

Combining Lemma \ref{convfor2} and Proposition \ref{levelsomme}, we get the following corollary.
\begin{corollary}\label{case2}
For $ r\leq n $, the minimum of $ E^{U}_{\mathcal{P}_{2}} $ when there are exactly $ r $ particles in the chains of $ \mathcal{P}_{2} $ is equal to the sum of the $ r $ smallest elements of $ \Gamma_{2} $.
\end{corollary}

 \subsection{Distribution of the energy levels}

By Corollary \ref{case2}, we need to understand the distribution of the energy levels in $ \Gamma_{2} $.
For $ \lambda>0 $, we define
\begin{equation}
 N_{2}^{U}(L,\lambda):=\frac{1}{L}\# \big\{  \, x\in \Gamma_{2}, \,  x\in (-\infty, \lambda] \big\}\quad  \text{ and }\quad N_{2}^{U}(\lambda):=\lim_{L\to \infty} N_{2}^{U}(L,\lambda).
\end{equation}
$ N_{2}^{U} $ is called the counting function of $ \Gamma_{2} $. We evaluate it in the following proposition.

\begin{proposition}\label{NUlambda}
Define the application $ J $ by, for $ \lambda>0 $,
\begin{align}\label{J}
J(\lambda)&:=(1-Me^{-l_{\rho,U}})^{2}\bigg(\int_{\mathcal{D}_{1}(\lambda)}e^{-u}\, du+\int_{\mathcal{D}_{2}(\lambda)}e^{-u}\,du \\
&\quad +\int_{0}^{M}\int_{\mathcal{D}_{3}(\lambda)}2 e^{-(u+v)}\,dtdudv +\int_{0}^{M}\int_{\mathcal{D}_{4}(\lambda,t)}2e^{-(u+v)}\,dtdudv \bigg).\nonumber
\end{align}
where
\begin{align*}
\mathcal{D}_{1}(\lambda)=\Big[\max\Big(l_{\rho,U},\frac{\pi}{\sqrt{\lambda}}\Big), 3l_{\rho,U}\Big]&, \qquad \mathcal{D}_{3}(\lambda)=\bigg\{(x,y)\in \big[l_{\rho,U},2l_{\rho,U}\big]^{2}, y\geq \max\Big(x,\frac{\pi}{\sqrt{\lambda}}\Big)\bigg\} \\
\mathcal{D}_{2}(\lambda)=\Big[ \max\Big(2l_{\rho,U},\frac{2\pi}{\sqrt{\lambda}}+\frac{\gamma}{8\pi^{2}}\Big),3l_{\rho,U}\Big]&, \qquad
\mathcal{D}_{4}(\lambda,t) =\bigg\{ (x,y)\in \big[l_{\rho,U},2l_{\rho,U}\big]^{2} ,\, y\geq x\geq \Big(\frac{\pi}{\sqrt{\lambda}}+\frac{\sigma(t)}{2y^{3}}\Big) \bigg\}.
\end{align*}
and $ \gamma $ (resp. $ \sigma(t) $) is given in Proposition \ref{Esolo} (resp. Proposition \ref{Epair}).

Then, with probability $ 1-O(L^{-\infty}) $, for every $ \beta>1 $ and $ \lambda>0 $, the counting function of $ \Gamma_{2} $ satisfies
$$ N^{U}_{2}(\lambda)=J(\lambda)+R_{\beta} $$
with $ R_{\beta}=O(\rho^{\beta}). $
\end{proposition}

\begin{proof}
A chain of $ \mathcal{P}_{2} $  is either a single piece $ \Delta_{i} $ or a pair $ \{\Delta_{i},\Delta_{j}\} $. In the first case, the energy levels of $ \Delta_{i} $ are functions of a single parameter, the length $ l_{i}\in [l_{\rho,U},3l_{\rho,U}] $. When $ I=\{\Delta_{i},\Delta_{j}\} $, the energy levels of $ I $ are given by the triplet of parameters $ (l_{i},l_{j},\textbf{d}_{ij})\in [l_{\rho,U},3l_{\rho,U}]\times[l_{\rho,U},3l_{\rho,U}]\times [0,M] $.

Fix $ \beta>1 $. We set a discretization of the above parameters with a constant step $ \rho^{\beta} $. We get a sequence of approximated energy levels $ \Gamma_{2}^{\beta} $. We prove that the Hausdorff distance between $ \Gamma_{2} $ and $ \Gamma_{2}^{\beta} $ is of order $ O(\rho^{\beta}) $. So it is sufficient to compute the counting function of $ \Gamma_{2}^{\beta} $ at order $ O(\rho^{\beta}) $. Since the Poisson process fix the statistics of pieces, one knows how many times each approximated energy level appears in $ \Gamma_{2}^{\beta} $. We will use the expansion of the energy levels given by Proposition \ref{Esolo} and Proposition \ref{Epair} to replace the condition "below $ \lambda $" by some conditions on the parameters.

We now give the details.
For $ I\in \mathcal{P}_{2} $, we distinguish two cases.
\begin{enumerate}
\item[(i)] If $ I=\{\Delta_{i}\} $ then $ l_{i}\in [k\rho^{\beta},(k+1)\rho^{\beta}) $ for some $ k  $ and for $ a\in \{1,2\} $, we approximate the $ a $-th energy level of the piece $ \Delta_{i} $ by
\begin{equation}
f_{a}(k)=f^{U}\big([0,k\rho^{\beta}],a\big) .
\end{equation}
The parameter $ k $ goes from $ K_{1}=\lfloor l_{\rho,U}\rho^{-\beta}\rfloor $ to $ K_{3}=\lfloor 3l_{\rho,U}\rho^{-\beta}\rfloor $. For $ a\in \{1,2\} $, we define 
\begin{equation}\label{nb4}
p_{a}(k)=\# \Big\{ \{\Delta_{i}\}\in \mathcal{P}_{2} , \, l_{i}\in [k\rho^{\beta},(k+1)\rho^{\beta}) \Big\};
\end{equation}
\item[(ii)] if $ I=\{\Delta_{j},\Delta_{k}\} $ then $ l_{j}\in [r\rho^{\beta},(r+1)\rho^{\beta}) $, $ l_{k}\in [s\rho^{\beta},(s+1)\rho^{\beta}) $ and $ \textbf{d}_{j,k}\in [d\rho^{\beta},(d+1)\rho^{\beta}) $ for some $ r,s $ and $ d $ and, for $ a\in \{1,2\} $, we approximate the $ a $-th energy level of the pair $ (\Delta_{j},\Delta_{k}) $ by
 \begin{equation}\label{nb5}
  g_{a}(r,s,d)=f^{U}\Big(\big\{[-r\rho^{\beta},0],[d\rho^{\beta},d\rho^{\beta}+s\rho^{\beta}]\big\},1\Big)
 \end{equation}
Here the parameters $ r $, $ s $ go from $ K_{1} $ to $ K_{2}=\lfloor 2l_{\rho,U}\rho^{-\beta}\rfloor $ and the parameter $ d $ goes from $ 0 $ to $ D=\lfloor M\rho^{-\beta}\rfloor $. For $ a\in \{1,2\} $, we set
\begin{equation}
q_{a}(r,s,d)=\# \Big\{ \{\Delta_{j},\Delta_{k}\}\in \mathcal{P}_{2},\, l_{j}\in [r\rho^{\beta},(r+1)\rho^{\beta}),\, l_{k}\in [s\rho^{\beta},(s+1)\rho^{\beta}),\, \textbf{d}_{j,k}\in [d\rho^{\beta},(d+1)\rho^{\beta})   \Big\}.
\end{equation}
\end{enumerate}
Let $ \Gamma_{2}^{\beta} $ to be the sequence of approximated energy levels.
\begin{lemma}\label{dinfty}
Recall the definition of the Hausdorff distance $ d_{\infty} $ on $ \mathcal{P}(\R) $. For $ (A,B)\in \mathcal{P}(\R)^{2} $,
$$ d_{\infty}(A,B):=\sup_{a\in A} \inf_{b\in B} \vert a-b\vert. $$
For $ \beta>1 $, there exists $ C>0 $ such that
$$ d_{\infty}(\Gamma_{2},\Gamma_{2}^{\beta})\leq C\rho^{\beta}. $$
\end{lemma}
\begin{proof}\textit{(of Lemma \ref{dinfty})}
By construction of $ \Gamma_{2}^{\beta} $, from $ x\in \Gamma_{2} $ we compute $ x^{\beta}\in\Gamma_{2}^{\beta} $. We study the cases separately.
\begin{enumerate}
\item[(i)] Either $ x^{\beta}=f_{1}(k) $, then $ x $ belongs to $ [f_{1}(k+1),f_{1}(k)] $. Note that
\begin{equation}\label{eq1}
 f_{1}(k)-f_{1}(k+1)=\frac{2\pi^{2}}{k^{3}\rho^{2\beta}}+O\Big(\frac{1}{k^{4}\rho^{2\beta}}\Big),
\end{equation}
\item[(ii)] Or $ x^{\beta}=f_{2}(k) $, then $ x $ belongs to $ [f_{2}(k+1),f_{2}(k)] $. Using (\ref{e(l,2)}), one  computes that
\begin{equation}\label{eq2}
 f_{2}(k)-f_{2}(k+1)= \frac{8\pi^{2}}{k^{3}\rho^{2\beta}}+O\Big(\frac{1}{k^{4}\rho^{2\beta}}\Big),
\end{equation}
\item[(iii)] Or $ x^{\beta}=g_{1}(r,s,d) $. Without lost of generality, assume that $ r<s $. Then $ x $ belongs to $ [g_{1}(r,s+1,d),g_{2}(r,s,d)] $. Using (\ref{e(j,k,1)}), one computes that
\begin{equation}\label{eq3}
g_{1}(r,s,d)-g_{1}(r,s+1,d)=\frac{2\pi^{2}}{s^{3}\rho^{2\beta}}+O\Big(\frac{1}{s^{4}\rho^{2\beta}}\Big),
\end{equation}
\item[(iv)] Or  $ x^{\beta}=g_{2}(r,s,d) $. Without lost of generality, assume that $ r<s $.  Then $ x $ belongs to $ [g_{2}(r+1,s,d),g_{2}(r,s,d)] $. Using (\ref{e(j,k,2)}), one computes that
\begin{equation}\label{eq4}
g_{2}(r,s,d)-g_{2}(r+1,s,d)=\frac{2\pi^{2}}{r^{3}\rho^{2\beta}}+O\Big(\frac{1}{r^{4}\rho^{2\beta}}\Big),
\end{equation}
\end{enumerate}
So
$$ \inf_{b\, \in \Gamma_{2}^{\beta}}\vert x-b \vert \leq C\frac{1}{r^{3}\rho^{2\beta}} $$
Since $ k $ (resp. $ r $ and $ s $) is of order $ O(l_{\rho,U}\rho^{-\beta}) $, we conclude
$$ \forall x\in \Gamma_{2} \qquad \inf_{b\, \in \Gamma_{2}^{\beta}}\vert x-b \vert \leq C\rho^{\beta} $$
\end{proof}
By Lemma \ref{nbUpsilon} and Lemma \ref{dinfty}, for $ \beta>1 $,
\begin{align}
\frac{1}{L}\Bigg\vert\# \bigg\{ x\in \Gamma_{2}, x\in (-\infty,\lambda]\bigg\}-\# \bigg\{ x\in \Gamma_{2}^{\beta}, x\in (-\infty,\lambda]\bigg\}\Bigg\vert&\leq \frac{\# \Gamma_{2}}{L}d_{\infty}(\Gamma_{2},\Gamma_{2}^{\beta})\\
&\leq C\rho^{\beta+1}\nonumber
\end{align}
Let $ N^{U}_{2,\beta} $ be the counting function of $ \Gamma_{2}^{\beta} $. Then, for $ \beta>1 $,
\begin{equation}\label{UUbeta}
N^{U}_{2}(\lambda)= N^{U}_{2,\beta}(\lambda)+O(\rho^{\beta+1}). 
\end{equation}

We estimate  $ N^{U}_{2,\beta} $ the counting function of $ \Gamma_{2}^{\beta} $. Set $ \lambda\in (\min \Gamma_{2}^{\beta},\max \Gamma_{2}^{\beta})  $. We translate the condition "energy level smaller than $ \lambda $" in term of bounds for the parameters of the discretization. For $ k\in \llbracket K_{1}, K_{3}-1 \rrbracket $,
\begin{equation}\label{condit1}
f_{1}(k)\leq \lambda \qquad
\Leftrightarrow \qquad  k\geq \frac{\pi}{\sqrt{\lambda}}\rho^{-\beta}
\end{equation}
Using the asymptotic (\ref{e(l,2)}), for large $ k $, we compute that
\begin{equation}\label{eq5}
f_{2}(k)=\frac{4\pi^{2}}{\Big(k-\frac{\gamma}{8\pi^{2}\rho^{\beta}}\Big)^{2}\rho^{2\beta}}+R_{k}
\end{equation}
with $ R_{k}=o(\frac{1}{k^{3}\rho^{2\beta}}). $

The remainder $ R_{k} $ is negligible with respect to the gap between $ f_{2}(k+1) $ and $ f_{2}(k) $ (see (\ref{eq2})). It yields
\begin{align}\label{condit2}
f_{2}(k)\leq \lambda \qquad
&\Leftrightarrow \qquad \frac{4\pi^{2}}{\Big(k-\frac{\gamma}{8\pi^{2}\rho^{\beta}}\Big)^{2}}\leq \lambda \rho^{2\beta}\\
&\Leftrightarrow \qquad k\geq \Big(\frac{2\pi}{\sqrt{\lambda}}+\frac{\gamma}{8\pi^{2}}\Big)\rho^{-\beta} \nonumber
\end{align}
For $ r,s \in \llbracket K_{1}, K_{2}-1 \rrbracket   $ and $ d\in \llbracket 0, D-1 \rrbracket $, assuming $ r\leq s $
\begin{equation}\label{condit3}
g_{1}(r,s,d)\leq \lambda \qquad
\Leftrightarrow \qquad s\geq \max\Big(r, \frac{\pi}{\sqrt{\lambda}}\rho^{-\beta}\Big)
\end{equation}
Using the asymptotic (\ref{e(j,k,2)}), for large $ r< s $ and $ d\in \llbracket 0, D\rrbracket $, we compute
\begin{equation}\label{eq6}
g_{2}(r,s,d)=\ \frac{\pi^{2}}{\Big(r-\frac{\sigma(d\rho^{\beta})}{2s^{3}\rho^{4\beta}}\Big)^{2}\rho^{2\beta}}+S_{r,s,d}
\end{equation}
with $ S_{r,s,d}=o(\frac{1}{r^{3}\rho^{2\beta}}) $.

The remainder $ S_{r,s,d} $ is negligible with respect to the gap between $ g_{2}(r+1,s,d) $ and $ g_{2}(r,s,d) $ (see (\ref{eq4})). It yields, for large $ r\leq s $,
\begin{align}\label{condit4}
g_{2}(r,s,d)\leq \lambda \qquad
&\Leftrightarrow \qquad \frac{\pi^{2}}{\Big(r-\frac{\sigma(d\rho^{\beta})}{2s^{3}\rho^{4\beta}}\Big)^{2}}\leq \lambda\rho^{2\beta}\\
&\Leftrightarrow \qquad r\geq \Big(\frac{\pi}{\sqrt{\lambda}}+\frac{\sigma(d\rho^{\beta})}{2s^{3}\rho^{3\beta}}\Big)\rho^{-\beta} \nonumber
\end{align}
$ \\ $
Thus, combining (\ref{condit1}), (\ref{condit2}), (\ref{condit3}) and (\ref{condit4}), for $ \lambda\in (\min \Gamma_{2}^{\beta}, \max \Gamma_{2}^{\beta} ) $,
\begin{align}\label{Ubetaa}
\# \bigg\{ x\in \Gamma_{2}^{\beta}, x\in (-\infty,\lambda]\bigg\}&=\sum_{k=k_{1}(\lambda)}^{K_{3}-1}p_{1}(k) + \sum_{k=k_{2}(\lambda)}^{K_{3}-1}p_{2}(k)+\sum_{d=0}^{D-1}\sum_{(r,s)\, \in\,  B(\lambda)}
\varepsilon(r,s) q_{1}(r,s,d) \\
&\quad +\sum_{d=0}^{D-1}\sum_{(r,s)\, \in\,  C(\lambda,d)}\varepsilon(r,s) q_{2}(r,s,d) \nonumber
\end{align}
where $ p_{a}(k) $ (resp. $ g_{a}(r,s,d) $) is given by (\ref{nb4}) (resp. (\ref{nb5})) and
\begin{align*}
k_{1}(\lambda)&:= \bigg\lceil\frac{\pi}{\sqrt{\lambda}}\rho^{-\beta}\bigg\rceil, \quad \qquad k_{2}(\lambda):=\bigg\lceil\Big(\frac{2\pi}{\sqrt{\lambda}}+\frac{\gamma}{8\pi^{2}}\Big)\rho^{-\beta}\bigg\rceil, \\
B(\lambda)&:=\bigg\{ (u,v)\in \llbracket K_{1},K_{2}-1\rrbracket^{2} ,\, v\geq \max\Big(u,\frac{\pi}{\sqrt{\lambda}}\rho^{-\beta}\Big) \bigg\}, \\
C(\lambda,d)&:=\bigg\{ (u,v)\in \llbracket K_{1},K_{2}-1\rrbracket^{2} ,\, v\geq u\geq \Big(\frac{\pi}{\sqrt{\lambda}}+\frac{\sigma(d\rho^{\beta})}{2v^{3}\rho^{3\beta}}\Big)\rho^{-\beta} \bigg\}, \\
\varepsilon(r,s)&:=2 \text{ if } r\neq s \text{ and } \varepsilon(r,s):=1 \text{ otherwise}.
\end{align*}
By Lemma \ref{nb3}, for $ \eta\in (\frac{2}{3},1) $, with probability $ 1-O(L^{-\infty}) $, we have for $ a\in \{1,2\} $ and for $ k,r,s,d $
\begin{align}
p_{a}(k)&=L\big(1-Me^{-l_{\rho,U}}\big)^{2}e^{-k\rho^{\beta}}(1-e^{-\rho^{\beta}})+r_{a}(k)L^{\eta} \\
q_{a}(r,s,d)&=L\big(1-Me^{-l_{\rho,U}}\big)^{2}e^{-(r+s)\rho^{\beta}}\rho^{\beta}\big(1-e^{-\rho^{\beta}}\big)^{2}+ s_{a}(r,s,d)L^{\eta}\nonumber
\end{align}
with $ r_{a}(k) $ and $ s_{a}(r,s,d) $ bounded for every $ k,r,s $ and $ d $.

Using dominated convergence theorem, we get
\begin{align}\label{NU2beta}
N^{U}_{2,\beta}(\lambda)&=\lim_{L\to \infty}\frac{1}{L}\# \bigg\{ x\in \Gamma_{2}^{\beta}, x\in (-\infty,\lambda]\bigg\}\\
&=\Big(1-Me^{-l_{\rho,U}}\Big)^{2}\bigg(\sum_{k=k_{1}(\lambda)}^{K_{3}-1}e^{-k\rho^{\beta}}(1-e^{-\rho^{\beta}})+\sum_{k=k_{2}(\lambda)}^{K_{3}-1}e^{-k\rho^{\beta}}(1-e^{-\rho^{\beta}})\nonumber\\
&\quad +\sum_{d=0}^{D-1}\sum_{(r,s)\, \in B(\lambda)}\varepsilon(r,s)e^{-(r+s)\rho^{\beta}}\rho^{\beta}(1-e^{-\rho^{\beta}})^{2}\nonumber\\
&\quad + \sum_{d=0}^{D-1}\sum_{(r,s)\, \in C(\lambda,d)}\varepsilon(r,s)e^{-(r+s)\rho^{\beta}}\rho^{\beta}(1-e^{-\rho^{\beta}})^{2} \bigg)\nonumber
\end{align}
Let $ \Sigma_{3} $ be the third sum in Equation (\ref{NU2beta}).
\begin{align}
\Sigma_{3}&:=\sum_{d=0}^{D-1}\sum_{(r,s)\, \in B(\lambda)}\varepsilon(r,s)e^{-(r+s)\rho^{\beta}}\rho^{\beta}(1-e^{-\rho^{\beta}})^{2}\\
&=\int_{0}^{D\rho^{\beta}}\bigg( \sum_{(r,s)\, \in B(\lambda)} \int_{r\rho^{\beta}}^{(r+1)\rho^{\beta}}\int_{s\rho^{\beta}}^{(s+1)\rho^{\beta}}2e^{-(u+v)}\,dudv\bigg)dt -D\rho^{\beta}(1-e^{-\rho^{\beta}})^{2} \sum_{r=k_{1}(\lambda)}^{K_{2}-1}e^{-2r\rho^{\beta}} \nonumber \\
&=\int_{0}^{D\rho^{\beta}}\int_{\mathcal{B}_{\beta}(\lambda)}2e^{-(u+v)}dudv - \frac{D}{2}\rho^{2\beta} e^{-2k_{1}(\lambda)\rho^{\beta}}\big(1+o(1)\big) 
\end{align}
where
\[\mathcal{B}_{\beta}(\lambda)=\bigg\{(x,y)\in \Big[K_{1}\rho^{\beta},K_{2}\rho^{\beta}\Big],\, y\geq \max\Big(x,\Big\lceil\frac{\pi}{\sqrt{\lambda}}\rho^{-\beta}\Big\rceil\rho^{\beta}\Big)\bigg\}. \]
Set
\[ \mathcal{B}(\lambda)=\bigg\{(x,y)\in [l_{\rho,U},2l_{\rho,U}]^{2}, y\geq \max\Big(x,\frac{\pi}{\sqrt{\lambda}}\Big)\bigg\}. \]
Using that, for any $ x>0 $,
\[ \Big\vert x-\lceil x\rho^{-\beta}\big\rceil \rho^{\beta} \big\vert \leq \rho^{\beta} \qquad \text{ and } \qquad \Big\vert x-\lfloor x\rho^{-\beta}\big\rfloor \rho^{\beta} \big\vert \leq \rho^{\beta} \]
we get
\begin{align}
\bigg\vert \, \Sigma_{3}-\int_{0}^{M}\int_{\mathcal{B}(\lambda)}2 e^{-(u+v)}\,dtdudv \,  \bigg\vert&\leq \rho^{\beta}\bigg( 2e^{-2l_{\rho,U}}\int_{\mathcal{B}(\lambda)}du+2M\int_{\mathcal{B}(\lambda)\backslash \mathcal{B}_{\beta}(\lambda)}du +Me^{-2l_{\rho,U}} \bigg) \\
&\leq \rho^{\beta}\bigg( 2e^{-2l_{\rho,U}}(2l_{\rho,U})^{2}+8M\rho^{\beta} +Me^{-2l_{\rho,U}} \bigg)\nonumber\\
&\leq C\rho^{\beta}.\nonumber
\end{align}

The other terms in Equation (\ref{NU2beta}) can be handled in much the same way.

So, for $ \lambda\in (\min\Gamma_{2}^{\beta},\max\Gamma_{2}^{\beta}) $,
\begin{align}\label{Ubeta}
N^{U}_{2,\beta}(\lambda)&=\Big(1-Me^{-l_{\rho,U}}\Big)^{2}\bigg(\int_{\frac{\pi}{\sqrt{\lambda}}}^{3l_{\rho,U}}e^{-u}\, du+\int_{\big(\frac{2\pi}{\sqrt{\lambda}}+\frac{\gamma}{8\pi^{2}}\big)}^{3l_{\rho,U}}e^{-u}\,du \\
&\quad +\int_{0}^{M}\int_{\mathcal{B}(\lambda)}2 e^{-(u+v)}\,dtdudv \nonumber \\
&\quad +\int_{0}^{M}\int_{\mathcal{C}(\lambda,t)}2e^{-(u+v)}\,dtdudv\bigg)+O(\rho^{\beta})\nonumber 
\end{align}
where 
\begin{align*}
\mathcal{B}(\lambda)&=\bigg\{(x,y)\in [l_{\rho,U},2l_{\rho,U}]^{2}, y\geq \max\Big(x,\frac{\pi}{\sqrt{\lambda}}\Big)\bigg\}\\
\mathcal{C}(\lambda,d)&=\bigg\{ (x,y)\in [l_{\rho,U},2l_{\rho,U}]^{2} ,\, y\geq x\geq \Big(\frac{\pi}{\sqrt{\lambda}}+\frac{\sigma(d)}{2y^{3}}\Big) \bigg\}.
\end{align*}
Combining (\ref{UUbeta}) and (\ref{Ubeta}), it yields
\begin{equation}\label{NU}
N^{U}_{2}(\lambda)=J(\lambda)+O(\rho^{\beta}).
\end{equation}
where $ J $ is given by (\ref{J}).
It concludes the proof of Proposition \ref{NUlambda}.
\end{proof}

The following corollary states that Assumption \ref{casequal} is true for $ p=2 $.
\begin{corollary}\label{corocasequal}
Set $ \delta\in (0,1) $. For every $ x\in\Gamma_{2} $ and in the thermodynamic limit,
\begin{equation}
\frac{1}{n}\#\big\{ y\in \Gamma_{2},\,  y=x \big\}=O(\rho^{2-\delta})
\end{equation}
\end{corollary}
\begin{proof}
Note that each domain of integration in the RHS of (\ref{J})  is smooth for $ \lambda\in (0,+\infty) $. So, $ J $ is continuous on  $ (\min \Gamma_{2},\max \Gamma_{2}) $. By Proposition \ref{NUlambda}, we compute for $ \beta >1 $, $ h>0 $ and $ x\in \Gamma_{2} $
\begin{align}
\frac{1}{n}\#\Big\{ y\in \Gamma_{2},\,  y=x \Big\}&\leq \frac{L}{n} \Big\vert N^{U}_{2}(L,x+h)-N^{U}_{2}(L,x-h)\Big\vert \\
&\leq \frac{L}{n}\bigg( \Big\vert N^{U}_{2}(L,x+h)-N^{U}_{2}(x+h)\Big\vert+\Big\vert N^{U}_{2}(L,x-h)-N^{U}_{2}(x-h)\Big\vert+\nonumber\\
& \qquad \Big\vert J(x+h)-J(x-h)\Big\vert +O(\rho^{\beta})\bigg)\nonumber\\
&\rightarrow_{\substack{L\to +\infty\\ \frac{n}{L}\to \rho}} \frac{1}{\rho}\Big\vert J(x+h)-J(x-h)\Big\vert +O(\rho^{\beta-1})
\end{align}
Taking $ \beta>2 $ and $ h\rightarrow 0 $, we conclude the proof of Corollary \ref{corocasequal}.
\end{proof}

\subsection{Construction of an approximated ground state}\label{approximatestate}

We use the counting function $ N^{U} $ to build an approximate ground state for $ H^{U}(\Lambda,n) $.

Note that, for $ d\in [0,M] $ and $ \min \Gamma_{2}<\lambda < \mu <\max \Gamma_{2} $,
\begin{equation}
\forall i\in \{1,2,3\} \quad \mathcal{D}_{i}(\lambda)\varsubsetneq \mathcal{D}_{i}(\mu )\qquad \text{ and }\qquad \mathcal{D}_{4}(\lambda,d)\varsubsetneq \mathcal{D}_{4}(\mu,d).
\end{equation}
So $ J $ is increasing on $ (\min \Gamma_{2},\max \Gamma_{2}) $. Remark also that, by Lemma \ref{nbUpsilon}, we have, for $ \lambda> \max \Gamma_{2} $, $ N^{U}_{2}(\lambda)>2\rho $ and, for $ 0<\lambda<\min \Gamma_{2} $, $ N^{U}_{2}(\lambda)=0 $.
Hence, by Proposition \ref{NUlambda} and the continuity of $ J $, for a fixed $ \beta>2 $, there exists a unique $ \lambda^{\beta}_{\rho}\in (\min \Gamma_{2}, \max \Gamma_{2}) $ such that $ J(\lambda^{\beta}_{\rho})=\rho-R_{\beta+1} $ or equivalently
\begin{equation}\label{NUkappa}
 N^{U}_{2}(\lambda^{\beta}_{\rho})=\rho.
\end{equation}
This unique $ \lambda_{\rho}^{\beta} $ is our \textit{Fermi energy level}. 

Consider all energy levels of $ \Gamma_{2} $ below $ \lambda_{\rho}^{\beta} $ and fill the chains by induction following the proof of the Proposition \ref{levelsomme}. Then, by definition, we get an occupation $ Q^{\beta} $ for which the number of particles in $\mathcal{P}_{2} $ is equal to $ n_{Q^{\beta}}=\min(n,LN^{U}_{2}(L,\lambda^{\beta}_{\rho})) $.
For $ L $ large enough (that depends on $ \rho $ and $ \beta $), 
\begin{equation}
\vert N^{U}_{2}(L,\lambda_{\rho}^{\beta})- N^{U}_{2}(\lambda_{\rho}^{\beta})\vert \leq \rho^{\beta+1}.
\end{equation}
So, using (\ref{NUkappa}), in the thermodynamic limit, the number of particles in the chains of $ \mathcal{N}_{2} $ is less than $ Cn\rho^{\beta} $ for some constant $ C>0 $.
Remembering $ \beta>2 $ and the left inequality of (\ref{encadrement}), for $ \rho $ small enough, one can set the restriction $ Q^{\beta}_{\vert \mathcal{N}_{2}} $ so that the occupation $ Q^{\beta} $ belongs to  $ \mathfrak{Q} $.

Set $ \delta_{\rho}^{\beta}=\frac{\pi}{\sqrt{\lambda_{\rho}^{\beta}}} $. Using Proposition \ref{NUlambda} and more specifically the R.H.S of (\ref{J}), one can get an approximate description of $ Q^{\beta} $ in term of the pieces' lengths and $ l_{\rho}^{\beta} $. Disregarding $ O(n\rho^{\beta}) $ particles, it means that
\begin{enumerate}
\item[$ \bigstar $] for a piece $ \Delta_{i} \in \mathcal{P}_{2} $
\begin{enumerate}
\item if $ l_{i}<\delta_{\rho}^{\beta} $, then $ q^{\beta}_{i}=0 $
\item if $ l_{i}\in \Big[\delta_{\rho}^{\beta},\, 2\delta_{\rho}^{\beta}+\frac{\gamma}{8\pi^{2}}\Big) $ then $ q^{\beta}_{i}=1 $ 
\item if $ l_{i}\geq 2\delta_{\rho}^{\beta}+\frac{\gamma}{8\pi^{2}} $ then $ q^{\beta}_{i}=2 $;
\end{enumerate} 
\item[$ \bigstar $] for a pair $ (\Delta_{j},\Delta_{k})\in \mathcal{P}_{2} $, assume $ l_{j}\leq l_{k} $
\begin{enumerate}
\item if $ l_{k}<\delta_{\rho}^{\beta} $ then $ q^{\beta}_{j}=q^{\beta}_{k}=0 $,
\item if $ l_{j}\in \Big[\delta_{\rho}^{\beta},\, \delta_{\rho}^{\beta}+\frac{\sigma(\textbf{d}_{j,k})}{2l_{k}^{3}}\Big) $ then $ q^{\beta}_{j}=0 $ and $ q^{\beta}_{k}=1 $
\item if $ l_{j}\geq \delta_{\rho}^{\beta}+\frac{\sigma(\textbf{d}_{j,k})}{2l_{k}^{3}} $ then  $ q^{\beta}_{j}=q^{\beta}_{k}=1 $
\end{enumerate}  
\end{enumerate}

We can compare the occupation $ Q^{\beta} $ with the occupation of the free operator $ Q^{0} $. Recall that in $ Q^{0} $ there are $ k $ particles in pieces of length between $ kl_{\rho} $ and $ (k+1)l_{\rho} $ where $ l_{\rho} $ is given by (\ref{lrho}). We compute
\begin{align}
\int_{\mathcal{D}_{1}(E_{\rho})}e^{-u}\, du+\int_{\mathcal{D}_{2}(E_{\rho})}e^{-u}\,du&=e^{-l_{\rho}}-e^{-3l_{\rho,U}}+e^{-2l_{\rho}-\frac{\gamma}{8\pi^{2}}}-e^{-3l_{\rho,U}}\\
&=\rho\Big(1-\rho+O(\rho^{2})\Big)\Big(1+e^{-\frac{\gamma}{8\pi^{2}}}\rho+O(\rho^{2})\Big),\nonumber \\
\int_{0}^{M}\int_{\mathcal{D}_{3}(E_{\rho})}2 e^{-(u+v)}\,dtdudv&=2M\int_{l_{\rho,U}}^{l_{\rho}}\int_{l_{\rho}}^{2l_{\rho,U}}e^{-(u+v)}\,dudv+2M\int_{2l_{\rho,U}\geq v\geq u\geq l_{\rho}}e^{-(u+v)}\,dudv\\
&=M\rho^{2}\big(1+O(\rho)\big),\nonumber \\ 
\int_{0}^{M}\int_{\mathcal{D}_{4}(E_{\rho},t)}2e^{-(u+v)}\,dtdudv& \leq 2M\int_{2l_{\rho,U}\geq v\geq u\geq l_{\rho}}e^{-(u+v)}\,dudv\\
&=M\rho^{2}\big(1+O(\rho)\big).\nonumber
\end{align}
So,
\begin{align}
N^{U}_{2}(E_{\rho})&\leq \rho\Big(1-2M\rho+O(\rho^{2})\Big)\bigg(\Big(1-\rho+O(\rho^{2})\Big)\Big(1+e^{-\frac{\gamma}{8\pi^{2}}}\rho+O(\rho^{2})\Big)+2M\rho\Big(1+O(\rho)\Big)\bigg)\\
&= \rho \Big(1+\rho\big(e^{-\frac{\gamma}{8\pi^{2}}}-1\big)+O(\rho^{2})\Big)\nonumber\\
&<\rho. \nonumber
\end{align}
Thus, $ E_{\rho}<\lambda_{\rho}^{\beta} $ meaning that $ l_{\rho,U}<\delta_{\rho}^{\beta}<l_{\rho} $. For $ \rho $ small enough, $ 2l_{\rho,U}+\frac{\gamma}{8\pi^{2}} \geq 2l_{\rho} $ so $ 2\delta_{\rho}^{\beta}+\frac{\gamma}{8\pi^{2}} \geq 2l_{\rho} $.
It means that when interactions are on, we remove one particle from pieces of length close to $ 2l_{\rho} $ but larger and put it in empty pieces of length close to $ l_{\rho} $ but smaller. Similarly, for pair of pieces of length close to $ l_{\rho} $, one takes one particle out of the pair to fill a smaller piece that does not interact.

$ \\ $

Hence, using (\ref{statefull}) and (\ref{statedecompchain}), we define the approximated ground state
\begin{equation}\label{betastate1}
\Psi^{\beta}(\Lambda,n)=\Psi^{U}(\Lambda,n,Q^{\beta}). 
\end{equation}

\begin{proposition}\label{thlimbeta}
Using the notations of Proposition \ref{NUlambda}, define the map $ \mathcal{J} $ by
\begin{align}\label{Jcurv}
\mathcal{J}(\lambda)&=L(1-Me^{-l_{\rho,U}})^{2}\bigg(\int_{\mathcal{D}_{1}(\lambda_{\rho}^{\beta})}f^{U}([0,u],1)e^{-u}\, du+\int_{\mathcal{D}_{2}(\lambda_{\rho}^{\beta})}f^{U}([0,u],2)e^{-u}\,du \\
&\quad +\int_{0}^{M}\int_{\mathcal{D}_{3}(\lambda_{\rho}^{\beta})}2 e^{-(u+v)}f^{U}(\{[-u,0],[t,v+t]\},1)\,dtdudv \nonumber \\
&\quad +\int_{0}^{M}\int_{\mathcal{D}_{4}(\lambda_{\rho}^{\beta},t)}2e^{-(u+v)}f^{U}(\{[-u,0],[t,v+t]\}2)\,dtdudv
\bigg).\nonumber
\end{align}
For $ \beta>2 $, for $ \lambda_{\rho}^{\beta} $ and $ \Psi^{\beta}(\Lambda,n) $ defined as above, for $ \delta\in (0,1) $ and $ 0<\rho<\rho_{\delta} $ small enough, then, in the thermodynamic limit, with probability $ 1 $,
\begin{equation}\label{energyapp}
\lim_{\substack{L\to +\infty\\ \frac{n}{L}\to \rho}}\frac{\big\langle H^{U}(\Lambda,n)\Psi^{\beta}(\Lambda,n),\Psi^{\beta}(\Lambda,n)\big\rangle}{n}=\frac{1}{\rho}\mathcal{J}(\lambda_{\rho}^{\beta}) +O(\rho^{2-\delta}).
\end{equation}
\end{proposition}

\begin{proof}
Fix $ \beta>2 $. By construction of $ \Psi^{\beta}(\Lambda,n) $ and using (\ref{energystatefull}), we write
\begin{equation}
\big\langle H^{U}(\Lambda,n)\Psi^{\beta}(\Lambda,n),\Psi^{\beta}(\Lambda,n)\big\rangle=E^{U}(\Lambda,n,Q^{\beta})= E^{U}_{\mathcal{P}_{2}}(Q^{\beta})+E^{U}_{\mathcal{N}_{2}}(Q^{\beta})
\end{equation}
By Proposition \ref{leftover}, we know that, for $ \delta\in (0,1) $ and $ \rho\in (0,\rho_{\delta}), $
\begin{equation}
E^{U}_{\mathcal{N}_{2}}(Q^{\beta})\leq n\rho^{2-\delta}
\end{equation}
It gives the amount of energy produced by particles we do not control precisely. One can check that it fits with the remaining part in (\ref{energyapp}).

Otherwise, we compute $ E_{\mathcal{P}_{2}}(Q^{\beta}) $ using $ \Gamma_{2}^{\beta} $, the approximate sequence of levels of energy for the good pieces that we introduced in the proof of Proposition \ref{NUlambda}. Following the method and the notations of Proposition \ref{NUlambda}, one derives the next formula. With probability $ 1-O(L^{-\infty}) $ and $ \eta\in (\frac{2}{3},1) $,
\begin{align}
 E_{\mathcal{P}_{2}}(Q^{\beta})
 &= L(1-Me^{-l_{\rho,U}})^{2}\bigg(\int_{\mathcal{D}_{1}(\lambda_{\rho}^{\beta})}f^{U}([0,u],1)e^{-u}\, du+\int_{\mathcal{D}_{2}(\lambda_{\rho}^{\beta})}f^{U}([0,u],2)e^{-u}\,du \\
&\quad +\int_{0}^{M}\int_{\mathcal{D}_{3}(\lambda_{\rho}^{\beta})}2 e^{-(u+v)}f^{U}(\{[-u,0],[t,v+t]\},1)\,dtdudv \nonumber \\
&\quad +\int_{0}^{M}\int_{\mathcal{D}_{4}(\lambda_{\rho}^{\beta},t)}2e^{-(u+v)}f^{U}(\{[-u,0],[t,v+t]\}2)\,dtdudv
\bigg) +O(L\rho^{\beta+1})+O(L^{\eta}).\nonumber
\end{align}

Thus, in the thermodynamic limit, one derives
\begin{align}
\lim_{\substack{L\to +\infty\\ \frac{n}{L}\to \rho}}\frac{\big\langle H^{U}(\Lambda,n)\Psi^{\beta}(\Lambda,n),\Psi^{\beta}(\Lambda,n)\big\rangle}{n}&=\frac{1}{\rho}\mathcal{J}(\lambda_{\rho}^{\beta})+O(\rho^{2-\delta})
\end{align}
It concludes the proof of Proposition \ref{thlimbeta}.
\end{proof}

\begin{remark}
One could also set
\begin{equation}\label{betastate2}
\Psi^{\beta}(\Lambda,n)=\bigg(\bigwedge_{I\, \in \,  \mathcal{P}_{2}}\psi^{U}\Big(I,(q_{i}^{\beta})_{i\in I}\Big)\bigg)\wedge \bigg(\bigwedge_{I\, \in \, \mathcal{N}_{2}}\bigwedge_{i\in I} \psi^{0}\Big(\Delta_{i},q_{i}^{\beta}\Big)\bigg)
\end{equation}
meaning that, outside of $ \mathcal{P}_{2} $, it behaves like a free state. By Remark \ref{remark1}, both states (\ref{betastate1}) and (\ref{betastate2}) give, up to the order $ O(\rho^{2-\delta}) $, the same amount of energy per particle in the thermodynamic limit.
\end{remark}

\subsection{Comparing the ground state energy to the approximated ground state energy}\label{compenergy}
We compare our approximate ground state energy with the ground state energy, in the thermodynamic limit.

\begin{proposition}\label{thlimU}
For $ L>0 $, let $ \Psi^{U}(\Lambda,n) $ be a ground state of $ H^{U}(\Lambda,n) $. For $ \delta\in (0,1) $ and $ \beta>3 $, the approximated ground state $ \Psi^{\beta}(L,n) $, given in Subsection \ref{approximatestate},  satisfies in the thermodynamic limit, with probability $ 1-O(L^{-\infty}) $,
\begin{equation}
\frac{\langle H^{U}(\Lambda,n)\Psi^{U}(\Lambda,n),\Psi^{U}(\Lambda,n)\rangle}{n}=\frac{\langle H^{U}(\Lambda,n)\Psi^{\beta}(L,n),\Psi^{\beta}(L,n)\rangle}{n}+O(\rho^{2-\delta}).
\end{equation} 
\end{proposition}
\begin{proof}
We drop the indices "$ \Lambda $" and "$ n $". 
Let $ \Psi^{U} $ be a ground state of $ H^{U} $. Using the notations of Subsection \ref{decompositiondelambda}, we have
\begin{equation}
\Psi^{U}=\sum_{Q\in \mathfrak{Q}}\lambda(Q)\psi^{U}_{\mathcal{P}_{2}}(Q)\wedge \psi^{U}_{\mathcal{N}_{2}}(Q)
\end{equation}
with $ \lambda(Q)\in \C $, $ \psi^{U}_{\mathcal{P}_{2}}(Q)\in \mathfrak{H}_{\mathcal{P}_{2}}(Q) $ and $ \psi^{U}_{\mathcal{N}_{2}}(Q)\in \mathfrak{H}_{\mathcal{N}_{2}}(Q) $,
Then,
\begin{align}
\langle H^{U}\Psi^{U},\Psi^{U}\rangle&=\sum_{Q\in \mathfrak{Q}} \vert\lambda(Q)\vert^{2}\Big(E^{U}_{\mathcal{P}_{2}}(Q)+E^{U}_{\mathcal{N}_{2}}(Q)\Big)\\
&\geq \min_{\mathfrak{Q}}E^{U}_{\mathcal{P}_{2}}\nonumber
\end{align}
Fix $ \beta>3 $ and $ \delta\in (0,1) $. Let $ \Psi^{\beta}=\Psi^{U}(Q^{\beta}) $ be the state given by the construction of Subsection \ref{approximatestate}. We compute
\begin{align}\label{ineqpsi}
0\leq \langle H^{U}\Psi^{\beta},\Psi^{\beta}\rangle-\langle H^{U}\Psi^{U},\Psi^{U}\rangle &\leq E^{U}_{\mathcal{P}_{2}}(Q^{\beta})+E^{U}_{\mathcal{N}_{2}}(Q^{\beta})-\min_{\mathfrak{Q}}E_{\mathcal{P}_{2}} \\
&\leq E^{U}_{\mathcal{P}_{2}}(Q^{\beta})-\min_{\mathfrak{Q}}E^{U}_{\mathcal{P}_{2}}+n\rho^{2-\delta} \nonumber
\end{align}
if $ \rho\in (0, \rho_{\delta}) $. We used Proposition \ref{leftover} for the last inequality.
If $ Q $ is an occupation that minimizes $ E_{\mathcal{P}_{2}} $ on $ \mathfrak{Q} $ then, by Proposition \ref{levelsomme},
\begin{equation}
 E^{U}_{\mathcal{P}_{2}}(Q^{\beta})-E^{U}_{\mathcal{P}_{2}}(Q)= \sum_{k=n_{Q}}^{n_{Q^{\beta}}}a_{k}.
\end{equation}
So,
\begin{equation}\label{enca1}
(\min \Gamma_{2}) \frac{n_{Q^{\beta}}-n_{Q}}{L} \leq \frac{E^{U}_{\mathcal{P}_{2}}(Q^{\beta})-E^{U}_{\mathcal{P}_{2}}(Q)}{L}\leq (\max \Gamma_{2} )\frac{n_{Q^{\beta}}-n_{Q}}{L}.
\end{equation}
By Lemma \ref{occupoutofchains},
\begin{equation}\label{enca2}
  0\leq \frac{n_{Q^{\beta}}-n_{Q}}{L}\leq \frac{n-n_{Q}}{L}\leq \frac{n\rho^{2-\delta}}{L}
\end{equation}
for $ \rho\in (0,\rho_{\delta}) $. Combining (\ref{enca1}) and (\ref{enca2}) we get
\begin{equation}\label{thlimit1}
\lim_{\substack{L\to +\infty\\ \frac{n}{L}\to \rho}} \frac{E^{U}_{\mathcal{P}_{2}}(Q^{\beta})-\min_{\mathfrak{Q}}E^{U}_{\mathcal{P}_{2}}}{n}=O(\rho^{2-\delta}).
\end{equation}
Thus, using (\ref{ineqpsi}) and (\ref{thlimit1}), one proves that, in the thermodynamic limit,
\begin{equation}
\lim_{\substack{L\to +\infty\\ \frac{n}{L}\to \rho}} \frac{\langle H^{U}\Psi^{\beta},\Psi^{\beta}\rangle-\langle H^{U}\Psi^{U},\Psi^{U}\rangle }{n}=O(\rho^{2-\delta})
\end{equation}
It concludes the proof of Proposition \ref{thlimU}.
\end{proof}
Combining Proposition \ref{thlimbeta}  and Proposition \ref{thlimU}, we get Theorem \ref{mainth}.

\subsection{Comparing a true ground state to the approximated ground state}

We recall that for $ \Psi \in \mathfrak{H}^{n}(\Lambda) $, we define its $ 1 $-particle density $ \gamma^{(1)}_{\Psi} $ (resp. $ 2 $-particle density $ \gamma^{(2)}_{\Psi} $) as the operator on $ \mathfrak{H}^{1}(\Lambda) $ (resp. $ \mathfrak{H}^{2}(\Lambda) $) given by (\ref{defgamma1}) (resp. \ref{defgamma2}). The following lemma deals with the case of a vector $ \Psi\in \mathfrak{H}^{n}(\Lambda) $ which factorizes with respect to a given partition of $ \Lambda $.

\begin{lemma}\label{gamma2}\cite{Klopp2020}
Consider $ (U_{i})_{1\leq i\leq r} $ a family of closed sets of $ \R $ where $ U_{i}\cap U_{j}=\emptyset $ holds for every $ i\neq j $ and $ \vert U_{i} \vert $ is finite. Set, for $ (q_{i})_{1\leq i\leq r}\in \N^{r} $,
\begin{equation}
\Psi=\bigwedge_{i=1}^{r}\psi(i,q_{i})
\end{equation}
where $ \psi(i,k) $ is a state that belongs to $ \mathfrak{H}^{k}(U_{i}) $, the $ k $-particle space on $ U_{i} $. Then the $ 1 $-particle $ \gamma^{(1)}_{\Psi} $ and the $ 2 $-particle $ \gamma^{(2)}_{\Psi} $ admit the following decompositions
\begin{equation}\label{gamma1pure}
\gamma^{(1)}_{\Psi}=\sum_{i=1}^{r}\gamma^{(1)}_{\psi(i,q_{i})}
\end{equation}
and
\begin{equation}\label{gamma2pure}
\gamma^{(2)}_{\Psi}=\sum_{i=1}^{r}\bigg(\gamma^{(2)}_{\psi(i,q_{i})}-\frac{1}{2}\,\gamma^{(1)}_{\psi(i,q_{i})}\otimes\, \gamma^{(1)}_{\psi(i,q_{i})}+\frac{1}{2}\Big(\gamma^{(1)}_{\psi(i,q_{i})}\otimes\,\gamma^{(1)}_{\psi(i,q_{i})}\Big)\circ \tau \bigg)+\frac{1}{2}\,\gamma^{(1)}_{\Psi}\otimes\,\gamma^{(1)}_{\Psi}-\frac{1}{2}\Big(\gamma^{(1)}_{\Psi}\otimes\, \gamma^{(1)}_{\Psi}\Big)\circ \tau
\end{equation}
with $ \tau(x_{1},x_{2},y_{1},y_{2})=(x_{1},x_{2},y_{2},y_{1}). $
\end{lemma}

We compare the $ 1 $-particle density and the $ 2 $-particle density of our approximate ground state with those of any ground state. The following Proposition is a reformulation of Proposition \ref{compgamma1} and Proposition \ref{compgamma2}.

\begin{proposition}\label{compgamma1et2}
Let $ \Psi^{U}(\Lambda,n) $ be a ground state of $ H^{U}(\Lambda,n) $. For $ \delta\in (0,1) $, $ \rho\in (0,\rho_{\delta})$ and $ \beta>3 $, set the approximate ground state $ \Psi^{\beta}(\Lambda,n) $ given in Subsection \ref{approximatestate}. Then, in the thermodynamic limit, with probability $ 1-O(L^{-\infty}) $, one has
\begin{equation}
\frac{1}{n}\Big\Vert \gamma^{(1)}_{\Psi^{U}(\Lambda,n)}-\gamma^{(1)}_{\Psi^{\beta}(\Lambda,n)} \Big\Vert_{\text{tr}} \leq 10\rho^{2-\delta}
\end{equation}
and
\begin{equation}
\frac{1}{n^{2}}\Big\Vert \gamma^{(2)}_{ \Psi^{U}(\Lambda,n)}-\gamma^{(2)}_{\Psi^{\beta}(\Lambda,n)}\Big\Vert_{\text{tr}}\leq 45\rho^{2-\delta}.
\end{equation}
\end{proposition}

\begin{proof}
Let $ \Psi^{U}(\Lambda,n) $ be a ground state of $ H^{U}(\Lambda,n) $ for large $ n $ and $ L $. The proof uses that both $ \Psi^{U} $ and $ \Psi^{\beta} $ admit a factor that fixes all but $ O(n\rho^{2-\delta}) $ particles. Indeed, by Lemma \ref{convfor2} and Corollary \ref{corocasequal}, both Assumption \ref{convex} and Assumption \ref{casequal} hold for $ p=2 $. So, we apply Proposition \ref{Qfix}. We have the factorization
\begin{equation}\label{decompPsi2}
\Psi^{U}=\bigg(\bigwedge_{I \in \mathcal{F}_{2}}\psi^{U}\big(I,(q_{i}^{\mathcal{F}_{2}})_{i\in I}\big)\bigg)\wedge \Omega^{U,\mathcal{F}_{2}^{c}}
\end{equation}
where
\begin{equation}
\Omega^{U,\mathcal{F}_{2}^{c}}=\sum_{Q\in \mathfrak{Q}}\lambda(Q)\bigwedge_{I \notin \mathcal{F}_{2}}\psi^{U}\big(I,(q_{i})_{i\in I}\big).
\end{equation}
Set
\begin{equation}
n^{\mathcal{F}_{2}}=\sum_{I\in \mathcal{F}_{2}}\sum_{i\in I}q_{i}^{\mathcal{F}_{2}}
\end{equation} the number of particles in $ \mathcal{F}_{2} $.

Let $ \Psi^{\beta} $ be our approximated ground state. By construction, we know
\begin{equation}
\forall I\in \mathcal{F}_{2} \quad \forall i\in I \qquad q_{i}^{\beta}=q_{i}^{\mathcal{F}_{2}}.
\end{equation}
 As in (\ref{decompPsi2}), we have
\begin{equation}
\Omega^{\beta,\mathcal{F}_{2}^{c}}=\bigwedge_{I\notin \mathcal{F}_{2}}\phi^{U}(I,(q_{i}^{\beta})_{i\in I})
\end{equation}
so that
\begin{equation}
\Psi^{\beta}=\bigg(\bigwedge_{I \in \mathcal{F}_{2}}\psi^{U}\big(I,(q_{i}^{\mathcal{F}_{2}})_{i\in I}\big)\bigg)\wedge \Omega^{\beta,\mathcal{F}_{2}^{c}}.
\end{equation}
We deal with the $ 1 $-particle densities and $ 2 $-particle densities separately.
\begin{enumerate}
\item[(i)]By Lemma \ref{gamma2}, the $ 1 $-particle density of $ \Psi^{U} $ satisfies
\begin{equation}
\gamma^{(1)}_{\Psi^{U}}=\sum_{I\in \, \mathcal{F}_{2}}\gamma^{(1)}_{\psi^{U}\big(I,(q_{i}^{\mathcal{F}_{2}})_{i\in I}\big)}+\gamma^{(1)}_{\Omega^{U,\mathcal{F}_{2}^{c}}}.
\end{equation}

For any $ \phi\in \mathfrak{H}^{n}(\Lambda) $, $ \vert \phi><\phi\vert $ is a rank one projector and
\begin{equation}
\Big\Vert  \vert \phi><\phi \vert \Big\Vert_{\text{tr}} = \int_{\Lambda^{n}} \vert \phi(X)\vert^{2}\,  dX.
\end{equation}
So its $ 1$-particle $ \gamma^{(1)}_{\phi} $ is trace class with
\begin{equation}
\Big\Vert \gamma^{(1)}_{\phi} \Big\Vert_{\text{tr}}=\int_{\Lambda}\gamma^{(1)}_{\phi}(x,x)\, dx 
\end{equation}
Since $ \Omega^{U,\mathcal{F}_{2}^{c}} $ is a normalized wave function of $ \mathfrak{H}^{n-n^{\mathcal{F}_{2}}}(\Lambda) $, we compute
\begin{equation}
\Big\Vert \gamma^{(1)}_{\Omega^{U,\mathcal{F}_{2}^{c}}} \Big\Vert_{\text{tr}}=n-n^{\mathcal{F}_{2}}\leq \max_{Q\in \mathfrak{Q}}\sum_{i\in \mathcal{N}}q_{i} + 2\#  \mathcal{P}_{2}\backslash \mathcal{F}_{2}\leq 5n\rho^{2-\delta}.
\end{equation}
Thus,
\begin{align}
\Big\Vert \gamma^{(1)}_{\Psi^{U}}-\gamma^{(1)}_{\Psi^{\beta}}\Big\Vert_{\text{tr}}&=\Big\Vert \gamma^{(1)}_{\Omega^{U,\mathcal{F}_{2}^{c}}} -\gamma^{(1)}_{\Omega^{\beta,\mathcal{F}_{2}^{c}}} \Big\Vert_{\text{tr}}  \\
&\leq \Big\Vert \gamma^{(1)}_{\Omega^{U,\mathcal{F}_{2}^{c}}} \Big\Vert_{\text{tr}}+\Big\Vert \gamma^{(1)}_{\Omega^{\beta,\mathcal{F}_{2}^{c}}} \Big\Vert_{\text{tr}}\nonumber \\
&\leq 10n\rho^{2-\delta}
\end{align}
\item[(ii)] We expand the $ 2 $-particle density of $ \Psi^{U} $ according to Lemma \ref{gamma2}.
 \begin{equation}
\gamma^{(2)}_{\Psi^{U}}=\gamma^{(2)}_{\Phi^{U,\mathcal{F}_{2}}}+\gamma^{(2)}_{\Omega^{U,\mathcal{F}_{2}^{c}}}+ \frac{1}{2}\Big( \gamma^{(1)}_{\Phi^{U,\mathcal{F}_{2}}}\otimes\gamma^{(1)}_{\Omega^{U,\mathcal{F}_{2}^{c}}}+\gamma^{(1)}_{\Omega^{U,\mathcal{F}_{2}^{c}}}\otimes\gamma^{(1)}_{\Phi^{U,\mathcal{F}_{2}}}-\big(\gamma^{(1)}_{\Phi^{U,\mathcal{F}_{2}}}\otimes\gamma^{(1)}_{\Omega^{U,\mathcal{F}_{2}^{c}}}\big)\circ \tau -\big(\gamma^{(1)}_{\Omega^{U,\mathcal{F}_{2}^{c}}}\otimes\gamma^{(1)}_{\Phi^{U,\mathcal{F}_{2}}} \big)\circ \tau  \Big).
 \end{equation}
 
For $ \phi\in \mathfrak{H}^{n}(\Lambda) $, the corresponding $ 2$-particle $ \gamma^{(2)}_{\phi} $ is trace class and it satisfies
\begin{equation}
\Big\Vert \gamma^{(2)}_{\phi} \Big\Vert_{\text{tr}}=\int_{\Lambda}\gamma^{(2)}_{\phi}(x_{1},x_{2},x_{1},x_{2})\, dx 
\end{equation} 
Then,
\begin{equation}
\Big\Vert \gamma^{(2)}_{\Omega^{U,\mathcal{F}_{2}^{c}}} \Big\Vert_{\text{tr}}=\frac{\big(n-n^{\mathcal{F}_{2}}\big)\big(n-n^{\mathcal{F}_{2}}-1\big)}{2}\leq \frac{25}{2}n^{2}\rho^{4-2\delta}
\end{equation}
and
\begin{equation}
\Big\Vert \gamma^{(1)}_{\Phi^{U,\mathcal{F}_{2}}}\otimes\gamma^{(1)}_{\Omega^{U,\mathcal{F}_{2}^{c}}} \Big\Vert_{\text{tr}}=\Big\Vert \Big(\gamma^{(1)}_{\Phi^{U,\mathcal{F}_{2}}}\otimes\gamma^{(1)}_{\Omega^{U,\mathcal{F}_{2}^{c}}}\Big)\circ \tau \Big\Vert_{\text{tr}}=n^{\mathcal{F}_{2}}\big(n-n^{\mathcal{F}_{2}}\big)\leq 5n^{2}\rho^{2-\delta}.
\end{equation}
The same inequalities hold for $ \Phi^{\beta,\mathcal{F}_{2}} $ and $ \Omega^{\beta,\mathcal{F}_{2}^{c}} $.
So,
\begin{equation}
\Big\Vert \gamma^{(2)}_{\Psi^{U}}-\gamma^{(2)}_{\Psi^{\beta}}\Big\Vert_{\text{tr}}\leq 45n^{2}\rho^{2-\delta}.
\end{equation}
\end{enumerate}

It concludes the proof of Proposition \ref{compgamma1et2}.
\end{proof}

$ \\ $

\section{Appendix}\label{appendix}

\subsection{Convex functions and discrete optimization}

\begin{definition}
A function $ F:\N \rightarrow \R $ is convex (resp. strictly convex) iff for every $ k\geq 1 $, $$ F(k+1)-F(k)\geq  F(k)-F(k-1) \qquad \text{(resp. } F(k+1)-F(k)> F(k)-F(k-1) \text{).} $$

\end{definition}

\begin{lemma}\label{convexsomme}
Let $ (F_{i})_{1\leq i\leq p} $ be nonnegative functions defined on $ \N $, with $ F_{i}(0)=0 $. Define 
\begin{equation}
F: \begin{cases}
 \qquad \N^{p} &\longrightarrow \quad \R_{+} \\ (x_{1},\dots, x_{p})&\longmapsto \sum_{i=1}^{p}F_{i}(x_{i})
\end{cases} \quad \text{ and } \quad G:\begin{cases} \, \N &\longrightarrow \quad \R_{+} \\ \, r&\longmapsto \min_{x_{1}+\dots+x_{p}=r}\,F(x_{1},\dots,x_{p}) \end{cases}
\end{equation}
Assume that, for every $ i $, $ F_{i} $ is strictly convex. Then,
\begin{enumerate}
\item  the function $ G $ is convex;
\item  for $ r\geq 1 $, $ G(r) $ is exactly the sum of the $ r $ smallest elements of $$ \Gamma=\big\{ F_{i}(k+1)-F_{i}(k), \, i\in \llbracket 1,m \rrbracket , \, k\in \N \big\}, $$ taken with multiplicity.
\end{enumerate}
\end{lemma}
\begin{proof}
\begin{enumerate}
\item For $ r\geq 1 $, choose $ (x_{1}^{r},\dots, x_{p}^{r})\in \N^{p} $ so that
$$ G(r)=F(x_{1}^{r},\dots, x_{p}^{r}). $$
We prove that one can set $ (x_{1}^{r+1},\dots, x_{p}^{r+1})\in \N^{p} $ satisfying
\begin{equation}
\exists ! j_{r+1}\in \llbracket 1, p \rrbracket \qquad \big( x_{j_{r+1}}^{r+1}=x_{j_{r+1}}^{r}+1 \big) \quad \text{and} \quad \big(\forall i\neq j_{r+1} \qquad x_{i}^{r+1}=x_{i}^{r}\big).
\end{equation}
Pick $ (y_{1},\dots ,y_{p})\in \N^{p} $ with $ \sum_{i=1}^{p}y_{i}=r+1 $. Assume that there is $ y_{i_{0}}>x_{i_{0}}^{r}+1 $. Without loss of generality we consider $ i_{0}=1 $.  Then
\begin{align*}
F(y_{1},\dots ,y_{p})-F(x_{1}^{r}+1,x_{2}^{r}\dots ,x_{p}^{r})=&F(y_{1}-1,y_{2},\dots, r+1-\sum_{i=1}^{p-1}y_{i})- F(x_{1}^{r},\dots, r-\sum_{i=1}^{p-1}x_{i}^{r}) \\
& +f_{1}(y_{1})-f_{1}(y_{1}-1)+f_{1}(x_{1}^{r})-f_{1}(x_{1}^{r}+1)\\
>&0
\end{align*}
by definition of $ (x_{i}^{r})_{1\leq i\leq p} $ and because $ f_{1} $ is strictly convex from $ 0 $ to $ r+1 $.

So $ x_{i}^{r+1}\leq x_{i}^{r}+1 $ for all $ i $. Since $ \sum_{i=1}^{p}x_{i}^{r+1}=\sum_{i=1}^{p}x_{i}^{r}+1 $, there is $ j_{0} $ so that $ x_{j_{0}}^{r+1}=x_{j_{0}}^{r}+1 $. Without loss of generality we can consider $ j_{0}=1 $. Pick $ (y_{1},\dots ,y_{p})\in \N^{p} $ with $ \sum_{i=1}^{p}y_{i}=r+1 $ and $ y_{1}=x_{1}^{r}+1 $. Then, the same calculus gives $ F(y_{1},\dots ,y_{p})\geq F(x_{1}^{r}+1, x_{2}^{r},\dots, x_{p}^{r}) $ meaning $ (x_{1}^{r}+1, x_{2}^{r},\dots, x_{p}^{r}) $ is a minimizer of $ F $. Thus we set $ (x_{i}^{r})_{r\geq 1} $ by induction and we compute
\begin{align*}
G(r+1)-G(r)&=f_{1}(x_{1}^{r}+1)-f_{1}(x_{1}^{r})\\
&> f_{1}(x_{1}^{r})-f_{1}(x_{1}^{r}-1)
\end{align*}
and for all $ j\in \llbracket 2, p\rrbracket $
$$ G(r+1)-G(r)\geq f_{j}(x_{j}^{r})-f_{j}(x_{j}^{r}-1) $$
because
$$ \sum_{i\notin \{ 1,j\}}f_{i}(x_{i}^{r})+f_{j}(x_{j}^{r})+f_{1}(x_{1}^{r})\leq \sum_{i\notin \{ 1,j\}}f_{i}(x_{i}^{r})+f_{j}(x_{j}^{r}-1)+f_{1}(x_{1}^{r}+1) $$
Hence,
$$ G(r+1)-G(r)\geq G(r)-G(r-1). $$
\item In particular, the sequence $ \big(G(r+1)-G(r)\Big)_{r\geq 0} $ is non decreasing and it belongs to $ \Gamma $. By \textit{reductio ad absurdum}, assume that there is $ a\in \Gamma\cap\big\{ G(r+1)-G(r),r\geq 1 \big\}^{c} $. Let $ r_{a} $ be such that $ G(r_{a})-G(r_{a}-1)\leq a < G(r_{a}+1)-G(r_{a}) $, and $ (i_{a}, x_{a}) $ such that $ a=F_{i_{a}}(x_{a}+1)-F_{i_{a}}(x_{a}) $. Then, $ x_{i_{a}}^{r_{a}}=x_{a} $ and
 $$ F(x_{1}^{r_{a}}, \dots, x_{a}+1, \dots , x_{p}^{r_{a}})=G(r_{a})+a<G(r_{a}+1). $$ Contradiction.
\end{enumerate}
It concludes the proof of Lemma \ref{convexsomme}.
\end{proof}

\subsection{Statistical distribution of the pieces} We recall some results about the statistical distribution of pieces.

\begin{proposition}\label{taillemax}\cite{Klopp2020}
With probability $ 1-O(L^{-\infty}) $, the largest piece has a length bounded by 
\newline
 $ \log(L)\log(\log(L)) $.
\end{proposition}

\begin{proposition}\label{nb1}\cite{Klopp2020}
Fix $ \beta\in (\frac{2}{3},1) $. For $ L $ large and $ a,b\in[0,\log(L)\log(\log(L))] $, with probability $ 1-O(L^{-\infty}) $ the number of pieces of length contained in $ [a,b] $ is equal to
$$ L(e^{-a}-e^{-b})+R_{L}L^{\beta} $$
where $ \vert R_{L} \vert $ is bounded.
\end{proposition}
\begin{proposition}\label{nb2}\cite{Klopp2020}
Fix $ \beta\in (\frac{2}{3},1)  $ and $ r\geq 2 $. For $ L $ large and $ (a_{i})_{1\leq i\leq r} ,(b_{i})_{1\leq i \leq r},(c_{i})_{1\leq i\leq r-1} $ and $ (d_{i})_{1\leq i \leq r-1} $ some positive sequences, with probability $ 1-O(L^{-\infty}) $, the number of pieces such that the length of $ i $-th piece (from left to right) is contained in $ [a_{i},b_{i}] $, the distance with the $(i+1) $-th piece is contained in $ [c_{i},d_{i}] $, is equal to
$$ L\prod_{i=1}^{r-1}(d_{i}-c_{i})\prod_{j=1}^{r}(e^{-a_{j}}-e^{-b_{j}})+R_{L}L^{\beta} $$
where $ \vert R_{L} \vert $ is bounded.
\end{proposition}
The proofs of Propositions \ref{taillemax}, \ref{nb1} and \ref{nb2} are in Appendix A of \cite{Klopp2020}.
From these propositions, we derive the following lemma.
\begin{lemma}\label{nb3}
Fix $ \beta\in (\frac{2}{3},1) $ and refer to the specific terminology in Definition \ref{chain}. For $ L $ large and $ a,b,c,d,f,g \in [l_{\rho,U},\log(L)\log\log(L)] $, with probability $ 1-O(L^{-\infty}) $,
\begin{enumerate}
\item the number of chains of size $ 1 $ with length contained in $ [a,b] $ is 
$$ L(1-Me^{-l_{\rho,U}})^{2}(e^{-a}-e^{-b})+S_{L}L^{\beta} $$
where $ \vert S_{L} \vert $ is bounded;
\item the number of chains of size $ 2 $ such that the length of the left piece is contained in $ [a,b] $, the length of the right piece is contained in $ [c,d] $ and the distance between the pieces is contained in $ [f,g] $, is equal to
$$ L(1-Me^{-l_{\rho,U}})^2(g-f)(e^{-a}-e^{-b})(e^{-c}-e^{-d})+S_{L}L^{\beta} $$
where $ \vert S_{L} \vert $ is bounded.
\end{enumerate} 
\end{lemma}
\proof

\begin{enumerate}
\item Let $ \mathcal{P}_{a,b}:= \{ i\in \llbracket 1,m \rrbracket, \, l_{i}\in [a,b] \}$. Then,
\begin{align*}
\{\text{chain of size } 1\}\cap \mathcal{P}_{a,b}= \mathcal{P}_{a,b}\backslash \Big( &\{i\in \mathcal{P}_{a,b},\,  \exists j>i, \, l_{j}\geq l_{\rho,U}, \, \textbf{d}_{i,j}\leq M \} \\ 
&\cup \{i\in \mathcal{P}_{a,b},\, \exists j<i, \, l_{j}\geq l_{\rho,U}, \, \textbf{d}_{j,i}\leq M \}\Big)
\end{align*}
We use Proposition \ref{nb1}, Proposition \ref{nb2} and  $ \# (A\cup B)=\# A +\# B -\#(A\cap B) $ to conclude.
\item Let $ \mathcal{R}_{a,b,c,d}:= \{ (i,j)\in \llbracket 1,m \rrbracket^{2}, \, i<j, \, l_{i}\in [a,b], \, l_{j}\in [c,d] \}$. Then,
\begin{align*}
 \{\text{chain of size } 2 \}\cap \mathcal{R}_{a,b,c,d}&=\mathcal{R}_{a,b,c,d}\backslash \Big( \{(i,j)\in \mathcal{R}_{a,b,c,d},\,  \exists k>j, \, l_{k}\geq l_{\rho,U}, \, \textbf{d}_{j,k}\leq M \} \\
& \qquad \qquad \qquad \cup \{i\in \mathcal{R}_{a,b,c,d},\, \exists k<i, \, l_{k}\geq l_{\rho,U}, \, \textbf{d}_{k,i}\leq M \}\Big)
\end{align*}
We conclude as for (1).
\end{enumerate}

\subsection{Bounds for the interaction of two particles}
\begin{lemma}\label{calculinteraction}
Set, for $ [a,b] \subset \R $ a finite interval and $ i\in \N $,
\[ \phi^{[a,b]}_{i}(x)=\frac{\sqrt{2}}{\sqrt{b-a}}\sin\big(\frac{\pi}{b-a} i(x-a)\big)\mathbf{1}_{[a,b]}(x).\]
For $ p,q\in \N $, if $ l>0 $ is large enough,
\begin{equation}
\int U(y-x)\Big\vert \phi^{[0,l]}_{p}\wedge \phi^{[0,l]}_{q}\Big\vert^{2}(x,y)dxdy\leq Cl^{-3}(p^{2}+q^{2})
\end{equation}
and if $ l'>0 $ is also large enough, and $ 0\leq d\leq M $,
\begin{equation}
\int U(y-x)\Big\vert \phi^{[-l',0]}_{p}\wedge \phi^{[d,d+l]}_{q}\Big\vert^{2}(x,y)dxdy\leq Cl^{-3}l'^{-3}p^{2}q^{2}
\end{equation}
with $ C>0 $ that only depends on $ U $.
\end{lemma}
\begin{proof}
We derive with changes of variables
\begin{align*}
\int U(y-x)\Big\vert \phi^{[0,l]}_{p}\wedge \phi^{[0,l]}_{q}\Big\vert^{2}(x,y)dxdy&=\int U(y-x)\bigg( \phi^{[0,l]}_{p}(x)^{2}\phi^{[0,l]}_{q}(y)^{2}-\phi^{[0,l]}_{p}(x)\phi^{[0,l]}_{q}(y)\phi^{[0,l]}_{p}(y)\phi^{[0,l]}_{q}(x)\bigg)dxdy\\
 &= 4l^{-1}\int_{-l}^{l}\int_{0}^{1}U(u)\bigg( \sin^{2}(\pi q (ul^{-1}+v))\sin^{2}(\pi pv) \\
 & \quad-\sin (\pi p(u l^{-1}+v))\sin(\pi q(ul^{-1}+v))\sin(\pi pv)\sin(\pi qv)\bigg) dudv \\
 &= 4l^{-3}\int_{-l}^{l}\int_{0}^{1}U(u) \pi^{2}u^{2}\bigg(q^{2}\cos^{2}(\pi pv)\sin^{2}(\pi q v) \\
 &\quad -pq \cos(\pi pv)\cos(\pi qv)\sin(\pi pv)\sin(\pi qv)\bigg)dudv + O(l^{-4})\\
 &\leq 10\pi^{2}l^{-3}(p^{2}+q^2)\int_{\R}U(u)u^{2}du
\end{align*}
and
\begin{align*}
\int U(y-x)\Big\vert \phi^{[-l',0]}_{p}\wedge \phi^{[d,d+l]}_{q}\Big\vert^{2}(x,y)dxdy&=\int U(x-y)\phi^{[-l',0]}_{p}(x)^{2}\phi^{[d,d+l]}_{q}(y)^{2}dxdy \\
&=4l^{-1}l'^{-1}\int_{0}^{l'}\int_{0}^{l}U(r+s+d)\sin^{2}(\pi prl'^{-1})\sin^{2}(\pi qsl^{-1})drds\\
&=4l^{-1}l'^{-1}\int_{0}^{+\infty}\int_{0}^{u}U(u+d)\sin^{2}\big(\pi p (u-v)l'^{-1}\big)\sin^{2}\big(\pi q v l^{-1})\text{d}u\text{d}v \nonumber \\
&= 4\pi^{4}l^{-3}l'^{-3}p^{2}q^{2}\int_{0}^{+\infty}\int_{0}^{u}U(u+d)(u-v)^{2}v^{2}\text{d}u\text{d}v \nonumber \\
&\quad+ O(l^{-4}l'^{-3}+l^{-3}l'^{-4})\nonumber\\
&\leq 8\pi^{4}l^{-3}l'^{-3}p^{2}q^{2}\int_{0}^{+\infty}U(u+d)u^{5}du.
\end{align*}
\end{proof}

\subsection{Proof of Proposition \ref{Epair}}  The ideas and the structure are inspired by the proof of Proposition \ref{Esolo} that one can find in Subsection 6.1.1 of \cite{Klopp2020}.

Set $ l>0 $, $ d>0 $ and $ a\geq 1 $. We consider the operator
\begin{equation}
\bigg(-\frac{d^{2}}{dy^{2}}^{D}_{\vert [-al,0]}\bigg)\otimes I +I\otimes \bigg(-\frac{d^{2}}{dx^{2}}^{D}_{\vert [d,d+l]}\bigg) + U(x-y) \qquad \text{ on } L^{2}([-al,0])\otimes L^{2}([d,d+l]) 
\end{equation}
By scaling, it is unitarily equivalent to the operator  $ l^{-2}H^{l} $ acting on $ L^{2}\big([0,1]^{2}\big) $ where
\begin{equation}
H^{l}=-\frac{\partial^{2}}{\partial x^{2}}-\frac{1}{a^{2}}\frac{\partial^{2}}{\partial y^{2}} + l^{2}U(lx+aly+d)
\end{equation}
with Dirichlet boundary conditions. Denote $ E^{l}_{0} $ the ground state of $ H^{l} $ and let $ H^{0} $ be the free operator. One checks that the eigenvalues of $ H^{0} $ are
\begin{equation} \label{eigv}
 E_{p,q}=\pi^{2}(p^{2}+q^{2}a^{-2})
\end{equation}
for $ p,q\geq 1 $, with the corresponding eigenfunctions 
\begin{equation} \label{eigf}
 \psi_{p,q}(x,y)=2\sin(\pi p x)\sin(\pi q y).
 \end{equation}
 
Set $ E_{0}:=E_{1,1} $, $ \psi_{0}:=\psi_{1,1} $ and $ U^{l}:=H^{l}-H^{0} $. By Lemma \ref{calculinteraction},
\begin{equation}\label{Upsi}
\langle\psi_{0}, U^{l}\psi_{0}\rangle\leq \frac{C}{a^{3}l^{4}}
\end{equation}
So,
\begin{equation}\label{reste}
 E_{0}\leq E^{l}_{0}\leq E_{0}+\langle \psi_{0}, U^{l}\psi_{0}\rangle\leq E_{0}+O(a^{-3}l^{-4})<E_{1,2}.
\end{equation}
Set $ \delta E=E^{l}_{0}-E_{0} $. By the Schur decomposition for $ (\text{Span}(\psi_{0}), \text{Span}(\psi_{0})^{\perp}) $, the eigenvalue equation becomes
\begin{equation}
\Pi_{0}U^{l}\Pi_{0}-(\delta E)\Pi_{0}  -\Pi_{0}U^{l}\Pi_{\perp}(H_{\perp}-E^{l}_{0})^{-1}\Pi_{\perp}U^{l}\Pi_{0}=0
\end{equation}
with $ \Pi_{0}=\vert \psi_{0}\rangle\langle \psi_{0}\vert $ the orthogonal projection on $ \text{Span}(\psi_{0}) $, $ \Pi_{\perp} $ the orthogonal projection on $ \text{Span}(\psi_{0})^{\perp} $ and $ H_{\perp}=H_{\vert \text{Span}(\psi_{0})^{\perp} } $. Note that $ H_{\perp}=\Pi_{\perp}H^{0}\Pi_{\perp}+\Pi_{\perp}U\Pi_{\perp}\geq E_{1,2} $.

We use the following notation
\begin{equation}\label{resperp}
R_{\perp}(z)=\Pi_{\perp}(H_{\perp}-z)^{-1}\Pi_{\perp}.
\end{equation}
We prove that one can replace $ R_{\perp}(E^{l}_{0})$ with $ R_{\perp}(E_{0})$ for some negligible cost. Remark that, as $ \Vert R_{\perp}(E_{0})\Vert\leq (E_{1,2}-E_{0})^{-1} $,
\begin{equation}\label{formres}
R_{\perp}(E^{l}_{0})-R_{\perp}(E_{0})=\sum_{n\geq1}R_{\perp}(E_{0})^{n+1}(\delta E)^{n}=O(\delta E)\leq Cl^{-4}.
\end{equation}
 Then,
 \begin{align*}
 \Big\vert \langle \psi_{0},U^{l}\big( R_{\perp}(E^{l}_{0})-R_{\perp}(E_{0})\big)U^{l}\psi_{0}\rangle\Big\vert&\leq \Big\Vert \sqrt{U^{l}}\Pi_{0}\Big\Vert^{2}\Big\Vert \sqrt{U^{l}}\big( R_{\perp}(E^{l}_{0})-R_{\perp}(E_{0})\big)\sqrt{U^{l}}\Big\Vert \\
 &\leq Cl^{-6}.
 \end{align*}
 using (\ref{Upsi}), (\ref{formres}) and $ \Vert U^{l} \Vert\leq l^{2}\Vert U \Vert_{\infty} $.
 
Thus, $ \delta E=A^{l}+O(l^{-6}) $ where
\begin{equation}
A^{l}=\langle \psi_{0},\big(U^{l}-U^{l}R_{\perp}(E_{0})U^{l}\big)\psi_{0}\rangle.
\end{equation} 
By (\ref{reste}), $ A^{l}=O(l^{-4}) $. 

Now we express $ R_{\perp}(z) $ in terms of $ R^{0}_{\perp}(z)=\Pi_{\perp}(H^{0}_{\perp}-z)^{-1}\Pi_{\perp} $. By Krein's formula, one can check that
\begin{equation}\label{res}
R_{\perp}(z)=\sqrt{ R^{0}_{\perp}(z)}\bigg(1+\sqrt{ R^{0}_{\perp}(z)}U^{l}\sqrt{R^{0}_{\perp}(z)}\bigg)^{-1}\sqrt{R^{0}_{\perp}(z)}.
\end{equation}
Denote 
\begin{equation}\label{renorm}
T^{l}=\sqrt{U^{l}}\sqrt{R^{0}_{\perp}(E_{0})} \quad \text{ and } \quad \phi^{l}_{0}=l^{2}\sqrt{U^{l}}\psi_{0}.
\end{equation}
Using (\ref{res}), one computes
\begin{align}\label{Al}
l^{4}A^{l}&=\langle \phi^{l}_{0},I-T^{l}\big(1+T^{l \star}T^{l}\big)^{-1}T^{l}\phi^{l}_{0}\rangle \nonumber\\
&=\langle \phi^{l}_{0},I-T^{l}T^{l\star}\big(1+T^{l}T^{l\star}\big)^{-1}\phi^{l}_{0}\rangle \nonumber\\
&=\langle \phi^{l}_{0},\big(1+T^{l}T^{l\star}\big)^{-1}\phi^{l}_{0}\rangle .
\end{align}

Define the partial isometry
$ \Gamma^{l}:\begin{cases} L^{2}\big([0,1]^{2}\big)&\longrightarrow L^{2}(\Omega) \\ \qquad f &\longmapsto \frac{1}{l\sqrt{a}}\mathbf{1}_{\Omega^{l}}(f\circ\gamma^{-1})\end{cases} $
where
\begin{equation}
\Omega= \Big\{ (u,v)\in \R_{+\star}^{2}, \, u> v \Big\} \quad \text{ and } \quad\Omega^{l}=\Big\{ (u,v)\in \R_{+\star}^{2}, \, \frac{u-v}{l}\in (0,1 ), \, \frac{v}{al}\in (0,1) \Big\}
\end{equation}
are two domains of $ \R^{2}_{+\star} $ and $ \gamma: \begin{cases} [0,1]^{2}&\longrightarrow \qquad\Omega \\ (x,y)&\longmapsto (l(x+ay),lay)\end{cases} $.

 Then, using (\ref{eigf}) and (\ref{renorm}), one computes
\begin{align}
\Gamma^{l}\phi^{l}_{0}(u,v)&=\frac{2l^{2}}{\sqrt{a}}\sqrt{U(u+d)}\sin\Big(\pi\frac{u-v}{l}\Big)\sin\Big(\pi \frac{v}{al}\Big)\mathbf{1}_{\Omega^{l}}(u,v)\nonumber\\
&=\frac{2\pi^{2}}{\sqrt{a}}\sqrt{U(u+d)}\Big( \frac{(u-v)v}{a}+\frac{(u-v)v}{l}g_{l,a}(u,v) \Big)\mathbf{1}_{\Omega^{l}}(u,v)
\end{align}
where $ g_{l,a} $ is a bounded continuous function. So, by dominated convergence theorem, the sequence $ \big(a^{3/2}\Gamma^{l}\phi^{l}_{0}\big)_{l>0} $ admits the following limit in $ L^{2}(\Omega) $ when $ l\to +\infty $:
\begin{equation}\label{limphi}
\varphi(u,v)=2\pi^{2}\sqrt{U(u+d)}(u-v)v. 
\end{equation} 

Otherwise, we use the notations (\ref{eigv}) and (\ref{eigf}) to write  the kernel $ K^{l} $ of $ \Gamma^{l}T^{l}T^{l\star}\Gamma^{l\star}  $.
$$ K^{l}(u,v,u',v')=\sum_{(p,q)\neq (1,1)}\frac{\mathbf{1}_{\Omega^{l}}(u,v)\mathbf{1}_{\Omega^{l}}(u',v')}{a(E_{p,q}-E_{0})}\sqrt{U(u+d)}\sqrt{U(u'+d)}\psi_{p,q}\Big(\frac{u-v}{l},\frac{v}{al}\Big)\psi_{p,q}\Big(\frac{u'-v'}{l},\frac{v'}{al}\Big). $$
If $ f\in \mathcal{C}_{c}^{\infty}(\Omega) $ then, for $ l $ large enough, $ f\in \mathcal{C}_{c}^{\infty}(\Omega^{l})  $ and
\begin{equation}\label{noy2}
 \Gamma^{l}T^{l}T^{l\star}\Gamma^{l\star}f(u,v)= \frac{4}{al^{2}}\sum_{(p,q)\neq (1,1)}\frac{\mathbf{1}_{\Omega^{l}}(u,v)}{\pi^{2}(\frac{p^{2}}{l^{2}}+\frac{q^{2}}{(al)^{2}})-\frac{E_{0}}{l^{2}}}\sqrt{U(u+d)}\sin\Big(\pi(u-v)\frac{p}{l}\Big)\sin\Big(\pi v\frac{q}{al}\Big)G_{f}\Big(\frac{p}{l},\frac{q}{al}\Big)
\end{equation}
where $G_{f}(\xi,\eta)=\int_{\Omega}\sqrt{U(u'+d)}f(u',v')\sin(\pi u'\xi)\sin(\pi v'\eta)\,du'dv' $.
By Riemann's summation, the limit for $ l\to +\infty $ of (\ref{noy2}) is
\begin{equation}\label{limL}
L(u,v)=\frac{4}{\pi^{2}}\mathbf{1}_{\Omega}(u,v)\sqrt{U(u+d)}\iint_{\R_{+}^{2}}\frac{1}{x^{2}+y^{2}} \sin(\pi(u-v)x)\sin(\pi v y)G_{f}(x,y)\, dxdy.
\end{equation}
Using $ g(s,t)=\sqrt{U(s+d)}f(s,t) $ and its Fourier transform $ \mathcal{F}_{g}(\xi,\eta)=\int_{\R_{+}^{2}}g(s,t)e^{is\xi+it\eta}\,dsdt $, one computes
$$ G_{f}(x,y)=\frac{1}{4}\Big(-\mathcal{F}_{g}(x,y-x)+\mathcal{F}_{g}(x,-y-x)-\mathcal{F}_{g}(-x,x-y)+\mathcal{F}_{g}(-x,x+y)\Big). $$
Then, (\ref{limL}) becomes
\begin{equation}\label{limL2}
 L(u,v)=-\frac{1}{\pi^{2}} \mathbf{1}_{\Omega}(u,v)\sqrt{U(u+d)}\iint_{\R^{2}}\frac{1}{x^{2}+y^{2}} \sin(\pi(u-v)x)\sin(\pi v y)\mathcal{F}_{g}(x,y-x)\, dxdy.
\end{equation}
\begin{lemma}\label{Sdef}
Define $ S $ on $ \mathcal{C}_{c}^{\infty}(\Omega) $ such that $ Sf=L $, given by (\ref{limL2}). Then, the operator $ S $ is well-defined and is extended to a bounded operator on $ L^{2}(\Omega) $.
\end{lemma}
\begin{proof}\textit{(of Lemma \ref{Sdef})} We first prove that, for $ (u,v)\in \Omega  $, $ L(u,v) $, given by (\ref{limL2}), is well-defined. We consider the singularities separately.
\begin{enumerate}
\item For $ (\alpha,\beta)\in \R^{2} $, we have
$$ \frac{1}{x^{2}+y^{2}}\sin(\alpha x)\sin(\beta y)\mathcal{F}_{g}(x,y-x)\sim_{(0,0)}\frac{xy}{x^{2}+y^{2}}\alpha\beta\mathcal{F}_{g}(0,0) $$
It gives the integrability in $ (0,0) $.
\item By the Paley Wiener theorem, as $ f\in \mathcal{C}_{c}^{\infty}(\Omega) $, $ \mathcal{F}_{g} $ is an entire function and $ \vert \mathcal{F}_{g}(x,y)\vert \leq \frac{C_{j}}{(1+\vert y \vert)^{j}} $ for $ j\geq 1 $. Then,
$$ \bigg\vert \frac{1}{x^{2}+y^{2}}\sin(\alpha x)\sin(\beta y)\mathcal{F}_{g}(x,y-x)\bigg\vert \leq \frac{C_{j}}{(x^{2}+y^{2})(1+\vert y-x \vert)^{j}}.  $$
It gives the integrability at $ \pm \infty $.
\end{enumerate}
So, $ S $ is well-defined on $ \mathcal{C}_{c}^{\infty}(\Omega) $.

Take $ h\in \mathcal{C}_{c}^{\infty}(\R^{2}) $. We compute
\begin{align*}
\int_{\R_{+}^{2}}U(u+v+d)\Big\vert \int_{\R^{2}}\frac{\sin(xu)\sin(yv)}{x^{2}+y^{2}}h(x,y)dxdy\Big\vert^{2}dudv&\leq \Vert h \Vert_{L^{2}}^{2}\int U(u+v+d)\frac{\sin^{2}(xu)\sin^{2}(yv)}{(x^{2}+y^{2})^{2}} \\
&\leq \Vert h \Vert_{L^{2}}^{2}\Bigg(\int U(u+v+d)\int_{[1,\infty)^{2}}\frac{1}{(x^{2}+y^{2})^{2}} \\
&\, + \int U(u+v+d)u^{2}v^{2}\int_{[0,1]^{2}}\frac{x^{2}y^{2}}{(x^{2}+y^{2})^{2}} \Bigg) \\
&\leq C\Vert h \Vert_{L^{2}}^{2}.
\end{align*}
Since the Fourier transform is unitary and $ U $ is bounded, we get that $ S $ admits an extension on $ L^{2}(\Omega) $.

It concludes the proof of Lemma \ref{Sdef}.
\end{proof}

Thus, by Lemma \ref{Sdef}, the sequence $ \big(\Gamma^{l}T^{l}T^{l\star}\Gamma^{l\star}\big)_{l>0} $ converges strongly to some operator $ S $. So does $ (1+\Gamma^{l}T^{l}T^{l\star}\Gamma^{l\star})^{-1} $ to $ (1+S)^{-1} $. The limit only depends on $ U $ and $ d $.

For any positive self-adjoint operator $ A $ on a Hilbert space $ \mathcal{H} $, we know $\Vert (1+A)^{-1}\Vert_{\mathcal{B}(\mathcal{H})}\leq 1. $
Then, combining it with (\ref{Al}), (\ref{limphi}) and (\ref{limL2}), for $ l $ large,
\begin{equation}
\langle a^{3/2}\Gamma^{l}\phi^{l}_{0},\big(1+\Gamma^{l}T^{l}T^{l\star}\Gamma^{l\star}\big)^{-1}a^{3/2}\Gamma^{l}\phi^{l}_{0}\rangle_{L^{2}(\Omega)}=\langle \varphi,(1+S)^{-1}\varphi\rangle_{L^{2}(\Omega)}+o(1).
\end{equation}
It yields
\begin{equation}
\delta E=\frac{1}{a^{3}l^{4}}\langle \varphi,(1+S)^{-1}\varphi\rangle_{L^{2}(\Omega)}+o\Big(\frac{1}{l^{4}}\Big).
\end{equation}
We set $ \tau(d)=\langle \varphi,(1+S)^{-1}\varphi\rangle_{L^{2}(\Omega)}. $
It concludes the proof of Proposition \ref{Epair}.

$ \\ $

\section*{Acknowlegdments}

The author would like to thank  very warmly his PhD supervisor Frederic Klopp for his guidance conceiving this article.

$ \\ $

\bibliographystyle{alpha}
\bibliography{Thermodynamic_limit_of_the_pieces_model}

\begin{thebibliography}{KPS19b}

\bibitem[AL18]{Alet2018}
Fabien Alet and Nicolas Laflorencie.
\newblock Many-body localization: An introduction and selected topics.
\newblock {\em C R Phys}, 19(6):498--525, 2018.

\bibitem[BW18]{Beaud2018}
V.~Beaud and S.~Warzel.
\newblock Bounds on the entanglement entropy of droplet states in the {XXZ}
  spin chain.
\newblock {\em Journal of Mathematical Physics}, 59(1):012109, jan 2018.

\bibitem[EKS18]{Elgart2018}
Alexander Elgart, Abel Klein, and Günter Stolz.
\newblock Many-body localization in the droplet spectrum of the random {XXZ}
  quantum spin chain.
\newblock {\em Journal of Functional Analysis}, 275(1):211--258, jul 2018.

\bibitem[KP21]{Kerner2021}
Joachim Kerner and Maximilian Pechmann.
\newblock {On the effect of repulsive pair interactions on
  Bose{\textendash}Einstein condensation in the Luttinger{\textendash}Sy
  model}.
\newblock {\em Proceedings of the American Mathematical Society},
  149(8):3499--3513, may 2021.

\bibitem[KPS19a]{Kerner2019a}
Joachim Kerner, Maximilian Pechmann, and Wolfgang Spitzer.
\newblock {{Bose-Einstein condensation in the Luttinger{\textendash}Sy model
  with contact interaction}}.
\newblock {\em Annales Henri Poincar{\'{e}}}, 20(6):2101--2134, feb 2019.

\bibitem[KPS19b]{Kerner2019}
Joachim Kerner, Maximilian Pechmann, and Wolfgang Spitzer.
\newblock {On Bose-Einstein condensation in the Luttinger{\textendash}Sy Model
  with finite interaction strength}.
\newblock {\em Journal of Statistical Physics}, 174(6):1346--1371, feb 2019.

\bibitem[KV20]{Klopp2020}
Fr\'{e}d\'{e}ric Klopp and Nikolaj~A. Veniaminov.
\newblock Interacting electrons in a random medium: a simple one-dimensional
  model.
\newblock In {\em Frontiers in analysis and probability}, pages 91--242.
  Springer, Cham, 2020.

\bibitem[LS73]{Luttinger1973}
J.~M. Luttinger and H.~K. Sy.
\newblock {Bose-Einstein condensation in a one-dimensional model with random
  impurities}.
\newblock {\em Physical Review A}, 7(2):712--720, feb 1973.

\bibitem[LZ06]{Lenoble2006}
Olivier Lenoble and Valentin Zagrebnov.
\newblock {Bose-Einstein Condensation in the Luttinger-Sy Model}.
\newblock April 2006.

\bibitem[Tes14]{Teschl2014}
Gerald Teschl.
\newblock {\em {Mathematical methods in quantum mechanics : with applications
  to Schrodinger operators}}.
\newblock American Mathematical Society, Providence, Rhode Island, 2014.

\bibitem[Ven12]{Veniaminov2012}
Nikolaj~A. Veniaminov.
\newblock {The Existence of the thermodynamic limit for the system of
  interacting quantum particles in random media}.
\newblock {\em Annales Henri Poincar{\'{e}}}, 14(1):63--94, may 2012.

\end{thebibliography}

\end{document}